\documentclass[10pt, two column, twoside]{IEEEtran}
\usepackage{amssymb}
\usepackage{amsmath} 
\usepackage[linesnumbered,lined,boxed,commentsnumbered, ruled]{algorithm2e}
\usepackage{mathrsfs}
\usepackage{bm}
\usepackage{tikz}
\usetikzlibrary{arrows}
\usepackage{subfigure}
\usepackage{graphicx,booktabs,multirow}
\usepackage{epstopdf}
\usepackage{subfigure}
\usepackage{multirow}
\usepackage{tabularx} 
\usepackage{booktabs}
\usepackage{float}

\definecolor{colorhkust}{RGB}{20,43,140}
\definecolor{colortsinghua}{RGB}{116,52,129}
\definecolor{color1}{RGB}{128,0,0}
\usepackage{color}

\usepackage{amsthm}

\newtheorem{lemma}{Lemma}
\newtheorem{theorem}{Theorem}

\newtheorem{proposition}{Proposition}
\newtheorem{definition}{Definition}

\newtheorem{fact}{Fact}

\newcommand{\mini}{\operatorname{minimize}}
\newcommand{\maxi}{\operatorname{maximize}}
\newcommand{\subj}{\operatorname{subject~to}}

\newcommand{\diagg}{\mathrm{diag}}
\newcommand\norm[1]{\|#1\|}

\newcommand{\rev}{\color{black}}

\begin{document}

        \title{Joint  Activity Detection and Channel Estimation for IoT Networks: Phase Transition and Computation-Estimation Tradeoff}
\author{Tao Jiang, \textit{Student Member}, \textit{IEEE}, Yuanming~Shi, \textit{Member}, \textit{IEEE},  
Jun~Zhang,~\IEEEmembership{Senior~Member,~IEEE,}
 and~Khaled~B.~Letaief,~\IEEEmembership{Fellow,~IEEE}
%        \thanks{This work was
%supported in part by the National Nature Science Foundation of China
%under Grant 61601290, and in part by the Shanghai Sailing Program under
%Grant 16YF1407700. }
        \thanks{T. Jiang and Y. Shi are with the School of Information Science and Technology, ShanghaiTech University, Shanghai 201210, China (e-mail: \{jiangtao1, shiym\}@shanghaitech.edu.cn).}
 \thanks{J. Zhang and K. B. Letaief are with the Department of Electronic and Computer Engineering, The Hong Kong University of Science and Technology, Hong Kong (e-mail:
\{eejzhang, eekhaled\}@ust.hk).}
        
}
        
        \maketitle
\begin{abstract}
Massive device connectivity is a crucial communication challenge for  Internet of Things (IoT) networks, which consist of a large number of devices with sporadic traffic. In each coherence block, the serving base station needs to identify the active devices and estimate their channel state information for effective communication. By exploiting the sparsity pattern of data transmission, we develop a structured group sparsity estimation method to simultaneously detect the active devices and estimate the corresponding channels. This method significantly reduces the signature sequence length while supporting massive IoT access. To determine the optimal signature sequence length, we study \emph{the phase transition behavior} of the group sparsity estimation problem. Specifically, user activity can be successfully estimated with a high probability when the signature sequence length exceeds a threshold; otherwise, it fails with a high probability. The location and width of the phase transition region are characterized via the theory of conic integral geometry. We further develop a smoothing method to solve the high-dimensional structured estimation problem with a given  limited time budget. This is achieved by sharply characterizing the convergence rate in terms of the smoothing parameter, signature sequence length and estimation accuracy, yielding a trade-off between the estimation accuracy and computational cost. Numerical results are provided to illustrate the accuracy of our theoretical results and the benefits of smoothing techniques.
\end{abstract}

\begin{IEEEkeywords}
Massive IoT connectivity, group sparsity estimation, phase transitions, statistical
dimension, conic integral geometry, and computation-estimation tradeoffs.
\end{IEEEkeywords}

\section{Introduction}
\label{sec:intro}
The explosion of {{small}} and cheap computing devices endowed with sensing and communication capability is paving the way towards the era of Internet of Things (IoT), which is expected to improve people's daily life and bring socio-economic benefits. For example, connecting the automation systems of intelligent buildings  to the Internet enables to control and manage different smart devices to save energy and improve the convenience for residents \cite{al2015internet}. Other applications include smart home, smart city and smart health care \cite{al2015internet}. To provide ubiquitous connectivity to enable such IoT based applications, massive machine-type communications and ultra-reliable and low latency communications become critical in the upcoming 5G networks \cite{dhillon2017wide,Wei_2018sparse}. In particular, in many scenarios, there are huge numbers of devices to be connected to the Internet via the base-station (BS). Thus supporting massive device connectivity is a crucial requirement for IoT networks \cite{de2017random,rajandekar2015survey,xu2018non}. 

Existing cellular standards, including 4G LTE \cite{ghosh2010fundamentals}, are unable to support massive IoT connectivity. Furthermore, the acquisition of the channel state information that is needed for the effective transmissions will bring huge overheads, and thus will make IoT communications even more challenging \cite{rajandekar2015survey}. Fortunately, the IoT data traffic is typically sporadic, i.e., only a few devices are active at any given instant out of all the  devices \cite{wunder2015sparse}. For example, in sensor networks, a device is typically designed to stay in the sleep mode and  is triggered only by external events in order to save energy. By exploiting the sparsity  in the device activity pattern, it  is  possible to design efficient schemes to support simultaneous device activity detection and channel estimation. As it is not feasible to assign orthogonal signature sequences to all the devices, this paper studies the \emph{Joint Activity Detection and channel Estimation (JADE)} problem considering non-orthogonal signature sequences \cite{dai2015non,chen2018sparse}.

\subsection{Related Work}\label{subsec:related}
A growing body of literatures have recently proposed various methods to deal
with  massive device connectivity and the high-dimensional channel estimation
problem. The compressed sensing (CS) based channel estimation techniques
have been proposed by exploiting the sparsity of channel structures in time, frequency,
angular and Doppler domains \cite{choi2017compressed,qin2018sparse, shen2016compressed}.
The spatial and temporal prior information was further exploited
to solve the high-dimensional channel estimation problem in dense wireless
cooperative networks \cite{liu2018massive}. However, in IoT networks
with a limited channel coherence time, it is critical to further exploit the
sparsity in the device activity pattern to enhance the channel estimation performance
\cite{Wei_2018sparse, chen2018sparse}, thereby reducing the training
overhead. Due to the large-scale nature of IoT communications,
it is also critical to develop efficient algorithms to address the computation
issue.

The sporadic  device activity detection problem has recently been  investigated.  {{\rev In the context of cellular networks, the random access scheme
 was investigated in \cite{hasan2013random,bjornson2017random} to deal with
the significant overhead  incurred by the massive number of devices.  In
the
random access scheme, a connection between an active device and the BS shall
be established if the orthogonal signature sequence randomly selected  by
the active device is not used by other devices. This scheme, however,  normally
causes collision among a huge number of devices. To support a massive number
of devices,  we thus focus on the non-orthogonal multiuser access  (NOMA)
  scheme \cite{dai2015non}, which is able to simultaneously serve multiple
devices via nonorthogonal
resource allocation. The opportunities and challenges of NOMA for supporting
massive connectivity are investigated in \cite{dai2015non}. Furthermore,
network densification \cite{shi2017generalized} turns to be a promising way
to improve network capacity,  enable low-latency mobile applications and
support massive device connectivity  by deploying more radio access points
in IoT
networks \cite{al2018small}} .} 

The information theoretical capacity for massive connectivity was studied in  \cite{chen2017capacity}. The sparsity activity pattern  yields a compressed sensing based formulation \cite{chen2018sparse, zhu2011exploiting} to detect the active devices and estimate the channels. {\rev Recall that the channel
state information (CSI) refers to the channel
propagation coefficients that describe how a signal propagates between transmitters
and receivers. In particular, in the related statements of ``prior knowledge
of CSI", CSI refers to the distribution information.}  Assuming perfect channel state information (CSI),  a sparsity-exploiting maximum a posteriori probability (S-MAP) criterion for multi-user detection in CDMA systems was developed in \cite{zhu2011exploiting}. The authors of \cite{schepker2015efficient,du2017efficient} considered the  multi-user detection problem with the aid of channel prior-information. {\rev In \cite{chen2018sparse,liu2018massive1,liu2018massive2},
a joint design of channel estimation and user activity detection via the
approximate
message passing (AMP) algorithm  was developed,  which leverages the statistical
channel information and large-scale fading coefficients to enhance the Bayesian
AMP algorithm with rigorous performance analysis. However, our approach
does not require prior information of the distribution of CSI to reduce the
signaling overhead}.  When assuming no prior knowledge of {\rev{the distribution of}} CSI, the joint user detection and channel estimation approach for  cloud radio access network via the ADMM  algorithm was proposed in \cite{he2017compressive} without performance analysis. 

In this paper, to eliminate the overheads of acquiring large-scale fading coefficients and statistical channel information, we propose a structured group  sparsity estimation approach to solve the JADE problem without prior knowledge of {\rev{the distribution of}} CSI. To determine the optimal signature sequence length, we provide {precise
characterization} for the phase transition behaviors in the structured group sparsity
estimation problem.
Although the bounds on the multi-user detection error in the non-orthogonal multiple access system have been presented in \cite{du2017efficient} based on the restricted isometry property \cite{foucart2013mathematical}, the order-wise estimates are normally not accurate enough for practitioners. A convex geometry approach was thus introduced in \cite{chandrasekaran2012convex} to  provide sharp estimates of the number of required  measurements for exact and robust recovery of structured signals. However, this approach can only provide the success conditions for signal recovery guarantees. Subsequently, the phase transition of a regularized linear inverse problem with random measurements was studied in \cite{amelunxen2014living,oymak2016sharp} based on the theory of conic integral geometry \cite{schneider2008stochastic}, which established both the success and failure conditions for signal recovery. In particular, the location and width of the transition are essentially controlled by the statistical dimension of a descent cone associated with the convex regularizers. However, these results are only applicable in the real domain. It is not yet clear how to apply the appealing methodology developed in \cite{amelunxen2014living} to provide sharp phase transition results for the high-dimensional estimation problem in the complex domain in IoT networks, which will be pursued in this paper. 

The large number of devices in IoT networks raises unique computational challenges when solving the JADE problem with a fixed time budget. Unfortunately, second-order methods like interior point method are inapplicable in large scale optimization problems due to its poor scalability. In contrast, first-order methods, e.g., gradient methods, proximal methods \cite{parikh2014proximal}, alternating direction method of multipliers (ADMM) algorithm \cite{boyd2011distributed,shi2015large}, fast ADMM algorithm \cite{goldstein2014fast} and Nesterov-type algorithms \cite{becker2011templates} are particularly useful for solving large-scale problems. Therefore, we focus on the first-order method in this paper. Furthermore, one way to minimize the computational complexity is to reduce the cost of each iteration by sketching approaches \cite{pilanci2015randomized,pilanci2016iterative}. However, this method is often suitable for solving an over-determined system instead of the under-determined linear system in our case. A different approach is to accelerate the convergence rate without increasing the computational cost of each iteration. It was shown in \cite{oymak2018sharp} that with more data it is possible to increase the step-size in the projected gradient method, thereby achieving a faster convergence rate. The authors of \cite{giryes2018tradeoffs} showed that by modifying the original iterations, it is possible to achieve faster convergence rates to maintain the estimation accuracy without increasing the computational cost of
each iteration considerably. More generally, smoothing techniques such as convex relaxation \cite{chandrasekaran2013computational} or simply adding a nice smooth function to smooth the non-differentiable objective function \cite{lai2013augmented,becker2011templates,bruer2015designing} often achieves a faster convergence rate. However, the amount of smoothing should be chosen carefully to guarantee the performance of sporadic device activity detection in IoT networks. In this paper, the smoothing method will be exploited to solve the high-dimensional group sparsity estimation problem with a fixed time budget by accelerating the convergence rate. This yields a trade-off between the computational cost and estimation accuracy, as increasing the smoothing parameter will normally reduce the estimation accuracy. The trade-off framework further provides guidelines for choosing the signature sequence length to maintain the estimation accuracy. 

{\rev 
        \subsection{Applications in IoT Systems}
        The proposed approach in this paper pervades a large number
of
applications in IoT systems. For instance,  detecting active devices shall
enhance data transmission efficiency in dynamic IoT networks \cite{al2017price}
and wireless sensor networks.  The proposed computation-estimation
trade-off techniques are particularly suitable for real-time wireless  IoT
networks, e.g., vehicular  networks \cite{al2018qos}, as well as providing
fault-tolerance communication and supporting high QoS and QoE requirements
\cite{hasan2017survey}  with low estimation
errors. While the lower computational complexity comes at the cost of relatively
high estimation errors, it shall reduce energy consumption significantly,
and thus is suitable for energy sensitive applications  \cite{mishra2018energy}.
 In addition,  the proposed approaches can be jointly designed with  the
secure access methods,
which shall enable smart applications of IoT devices especially related to
healthcare applications \cite{al2018context}. 
}

\subsection{Contributions}
The major contributions of the paper are summarized as follows:
\begin{itemize}
        \item 
        By exploiting sparsity  in the device activity pattern, we propose a structured group sparsity estimation approach to solve the JADE problem for massive IoT connectivity. Our method is widely applicable and does not depend on the knowledge of channel statistical information and the large-scale fading coefficients.  
        \item 
        Based on the theory of conic integral geometry,
        we provide precise prediction for the location and the width of the phase
        transition region of the sparsity estimation problem via establishing both the failure  and success conditions for signal recovery.  This result provides theoretical guidelines for choosing the optimal signature sequence length to support massive IoT connectivity and channel estimation.
        We also provide evidence that massive {\rev multiple input multiple output (MIMO) system}  is particularly suitable for supporting massive IoT connectivity, as the width of the phase transition region can be narrowed to zero asymptotically as the number of BS antennas increases.

        \item We further contribute this work by computing the statistical dimension for the descent cone of the  group sparsity inducing regularizer to determine the phase transition of the high-dimensional group sparsity estimation problem. The success of this work is based on the proposal of transforming the original complex estimation problem into the real domain, thereby leveraging the theory of conic integral geometry.    

        \item To solve the high-dimensional group sparsity estimation problem with a fixed time budget, we adopt the smoothing method to smooth the non-differentiable group sparsity inducing regularizer to accelerate the convergence rates. We further characterize
the sharp trade-offs between the computational cost and estimation accuracy. This helps guide the signature sequence design to maintain
the estimation accuracy for the smoothed estimator.  Numerical results shall be provided to show the benefits
of smoothing techniques.

\end{itemize}

\emph{Notations}: Uppercase/lowercase boldface letters denote matrices/vectors. For an $L\times 2N$ matrix $\bm Q$, we denote its $i^{th}$ row by $\bm q^{i}$ , its $j^{th}$ column by $\bm q_{j}$. Let $\bm Q_{\mathcal{V}_i}=[(q^i)^T,(q^{i+N})^T]^T$ denote the row submatrix of $\bm Q $ consisting of the rows indexed by $ \mathcal{V}_i=\{i,i+N \} $. The operator $\norm{\cdot}_2, \norm{\cdot}_F, (\cdot)^T,\Re(\cdot),\Im(\cdot) $ stand for transpose, Euclidian norm, Frobenius norm, real part, imaginary part.
$\bm Q\thicksim\mathcal{CN}(\mu,\sigma^2\bm I)$ denotes that each element in $\bm Q$ follows i.i.d. normal distribution with mean $\mu$ and variance $\sigma^2$.

\section{System model and problem formulation}\label{sys}
\subsection{System Model and Problem Formulation}
We consider an IoT network with one BS serving $N$ single-antenna IoT devices, where the BS is equipped with $M$ antennas. The channel vector from device $i$ to the BS is denoted by $\bm{h}_i\in\mathbb{C}^{ M}$, $i=1,\cdots,N$. With sporadic communications, only a few devices are active out of all devices \cite{wunder2015sparse} as shown in Fig.\ref{fig:sys_model}. We consider the synchronized wireless system with block fading. That is, each device is active during a coherence block, and is inactive otherwise.  In each block, we define the device activity indicator as follows: $a_i=1$ if device $i$ is active, otherwise $a_i=0$. Furthermore, we define the set of active devices within a coherence block as $\mathcal{S}=\left\{i| a_i=1, i=1,\cdots,N\right\}$ with $\left|\mathcal{S}\right|$ denoting the number of active devices.

\begin{figure}[htb]
        \begin{minipage}[b]{1\linewidth}
                \centering
                \centerline{\includegraphics[width=0.91\columnwidth]{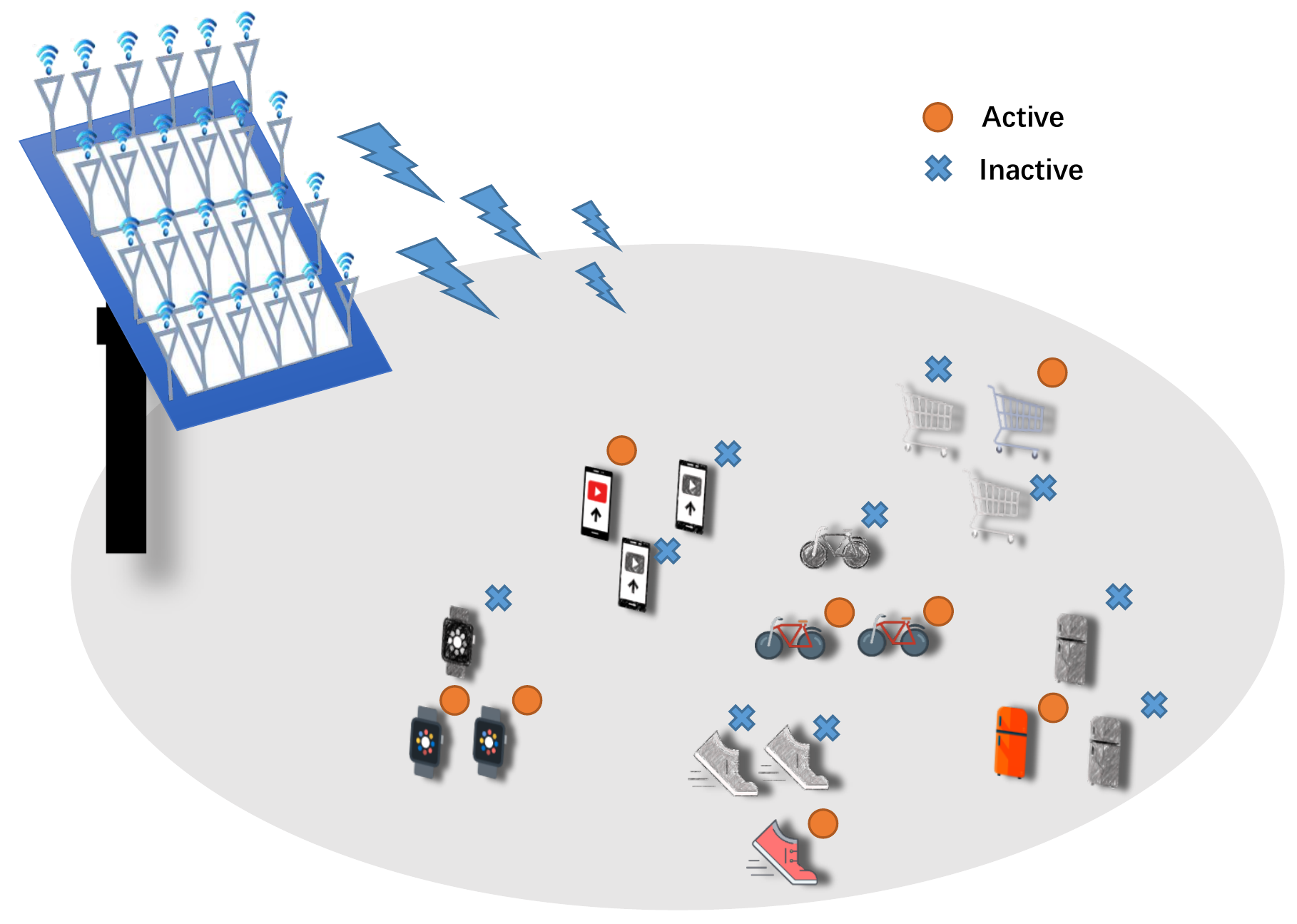}}
                %  \vspace{1.5cm}
        \end{minipage}
        \caption{A typical IoT network with massive sporadic traffic devices.} \label{fig:sys_model}
\end{figure}

For uplink transmission in a coherence block with length $T$, we consider the {\it{Joint Activity Detection and channel Estimation (JADE)}} problem. Specifically, the received signal at the BS is given by
\begin{equation}\label{eq:signal}
        \bm{y}(\ell)=\sum_{i=1}^{N}\bm{h}_i a_i q_i(\ell)+\bm n(\ell) =\sum_{i\in\mathcal{S}}\bm{h}_iq_i(\ell)+\bm n(\ell),
\end{equation}
for all $\ell=1,\dots, L$. Here, $L<T$ is the  length of the signature sequence, $q_i(\ell)\in\mathbb{C}$ is the signature symbol transmitted from device $i$ at time slot $\ell$,  $\bm{y}(\ell)\in\mathbb{C}^{ M}$ is the received signal at the BS, and $\bm{n}(\ell)\in\mathbb{C}^{ M}$ is the additive noise distributed as $\mathcal{CN}(\bm 0,\sigma^2\bm I)$.

With massive devices and a limited channel coherence block, the length of the signature sequence is typically smaller than the total number of devices, i.e., $L \ll N$. It is thus impossible to assign mutually orthogonal sequences to all the devices. As suggested in \cite{chen2018sparse}, we generate the signature sequences from i.i.d. complex Gaussian distribution with zero mean and variance one, i.e., each device $i$ is assigned a unique signature  sequence  $q_i(\ell)\thicksim\mathcal{CN}(0,1),~ \ell=1,\cdots,L $. Notice these sequences are non-orthogonal.

Let $\bm{Y}=[\bm{y}(1),\dots, \bm{y}(L)]^{ T}\in \mathbb{C}^{L\times M}$ denote the received signal across $M$ antennas, $\bm{H}=[\bm{h}_1,\dots, \bm{h}_N]^{ T}\in\mathbb{C}^{N\times M}$ be the channel matrix from all the devices to the BS antennas, and $\bm{Q}= [\bm{q}(1),\dots, \bm{q}(L)]^{ T}\in\mathbb{C}^{L\times N}$ be the known signature matrix with $\bm{q}(\ell)=[q_1(\ell),\dots, q_{N}(\ell)]^ {T}\in\mathbb{C}^{N}$. We  rewrite \eqref{eq:signal} as
\begin{equation}\label{eq:sys_model}
        \bm{Y}=\bm{Q}\bm{A}\bm{H}+\bm{N},
\end{equation}
where $\bm{A}={\rm{diag}}(a_1,\dots, a_n)\in\mathbb{R}^{N
        \times N}$ is the diagonal activity matrix and $\bm{N}=[\bm{n}(1),\dots,
\bm{n}(L)]\in\mathbb{C}^{L\times M}$ is the additive noise matrix.
Our goal  is to jointly estimate the channel matrix $\bm{H}$ and detect the activity matrix $\bm{A}$.

Let $\bm{\Theta}^{0}=\bm{A}\bm{H}\in\mathbb{C}^{N\times M}$ with $\bm{A}$ as the sparse diagonal activity matrix. Matrix $\bm{\Theta}_0$ thus has the structured group sparsity pattern in its rows \cite{Wainwright2014structured}. The linear measurement model \eqref{eq:sys_model} can be further rewritten as
\begin{equation}\label{eq:formulation}
        \bm{Y}=\bm{Q}\bm{\Theta}_0+\bm{N}.
\end{equation}

To estimate the group row sparse matrix $\bm{\Theta}_0$, we introduce the following convex group sparse inducing norm (i.e., mixed $\ell_1/\ell_2$-norm) in the form of \cite{Wainwright2014structured}
\begin{equation}
        \mathcal{R}(\bm{\Theta}):=\sum_{i=1}^{N}\norm{\bm\theta^i}_2,
\end{equation}
where $\bm{\theta}^i\in\mathbb{C}^{1\times M}$ is the $i$-th row of matrix $\bm \Theta$. This norm will help to induce a group sparsity structure in the solution. The resulting group sparse matrix estimation problem, i.e., the JADE problem, can thus be formulated as the following convex optimization problem:
\begin{equation}\label{eq:formulation}
        \begin{split}
                \mathscr{P}:~
                &\underset{\bm \Theta\in\mathbb{C}^{N\times M}}{\mini}
                ~ \mathcal{R}(\bm{\Theta})\\
                &{\subj}~ \norm{\bm{Q\Theta}-\bm Y}_F\le \epsilon,
        \end{split}
\end{equation}
where $ \epsilon $ is an upper bound on $ \norm{\bm N}_F $ and assumed to be known
as a \emph{priori}.  Given the estimate matrix ${\bm{\hat{\Theta}}}$, the activity matrix can be recovered as ${\bm{\hat{A}}}={\diagg}(\hat{a}_1,\cdots,\hat{a}_{n})$, where $\hat{a}_i=1$ if $\norm{\bm{\hat\theta}^i}_2\geq \gamma_0$ for a small enough threshold $\gamma_0 (\gamma_0\geq0)$; otherwise, $\hat{a}_i=0$. The estimated channel matrix for the active devices is thus given by $\bm{\hat{H}}$ with its $i$-th row as $\bm{\hat{h}}^i=\bm{\hat{\theta}}^i$ where $i\in\{j|\hat{a}_j=1 \}$.  

\subsection{Problem Analysis}
\subsubsection{Phase Transitions}Due to the limited radio resources, it is critical to \emph{precisely} find the minimal number of signature symbols $L$ to support massive device access. This can be achieved by precisely revealing the locations of the \emph{phase transition} of the high-dimensional group sparsity estimation problem via solving the convex optimization problem $ \mathscr{P}$. Although recent years have seen progresses on structured signal estimation \cite{donoho2006compressed, candes2006robust, chandrasekaran2012convex}, they only provide a success condition for signal recovery without precise phase transition analysis. The recent work  \cite{amelunxen2014living} provided a principled framework to predict phase transitions (including the location and width of the transition region) for random cone programs  \cite{hu2018performance} via the theory of conic integral geometry. Unfortunately, the approach based on conic integral geometry is only applicable in the real field case, which thus cannot be directly applied for problem $\mathscr{P}$ in the complex field. To address this issue, we  propose to approximate the original complex estimation problem $\mathscr{P}$ by a real estimation problem, followed by precise phase transition analysis via conic integral geometry \cite{amelunxen2014living}. Theoretical results and numerical experiments
will  provide evidences that the approximations are quite tight.  We shall prove that the locations of phase transitions are determined by the intrinsic geometry invariants (i.e., the statistical dimension) associated with the high-dimensional estimation
problem $\mathscr{P}$. In particular,
we will show that the width of the transition region can be reduced to zero asymptotically in the
limit as the number of antennas at the BS goes to infinity.
Therefore, massive MIMO is especially well-suited for supporting massive IoT
connectivity by providing accurate phase transition location.

\subsubsection{Computation and Estimation Trade-offs} To address the computational challenges in massive IoT networks with a limited time budget, we adopt the smoothing method to smooth the non-differentiable group sparsity inducing regularizer to accelerate the convergence rates. The computational speedups can be achieved by projecting onto simpler sets \cite{chandrasekaran2013computational},
varying the amount of smoothing \cite{bruer2015designing}, or adjusting the step
sizes \cite{oymak2018sharp} applied to the optimization algorithms. However, the computational speedups will normally reduce the estimation accuracy. Based
on the phase transition results, we shall propose to control the
amount of smoothing to achieve sharp computation and estimation tradeoffs
for the smoothed optimization problem $\mathscr{P}$ via the smoothing method. The smoothed formulation
can be further efficiently solved via various efficient first-order methods with cheap iterations and low memory cost, e.g., gradient methods, proximal methods \cite{parikh2014proximal}, alternating direction method of multipliers (ADMM) algorithm \cite{boyd2011distributed}, fast ADMM algorithm \cite{goldstein2014fast} and Nesterov-type algorithms
\cite{becker2011templates}.

\section{Precise Phase transition analysis}
In this section, we study the phase transition phenomenon when solving the JADE problem. 
\begin{figure}[htb]
        \begin{minipage}[b]{1\linewidth}
                \centering
                \centerline{\includegraphics[width=0.96\columnwidth]{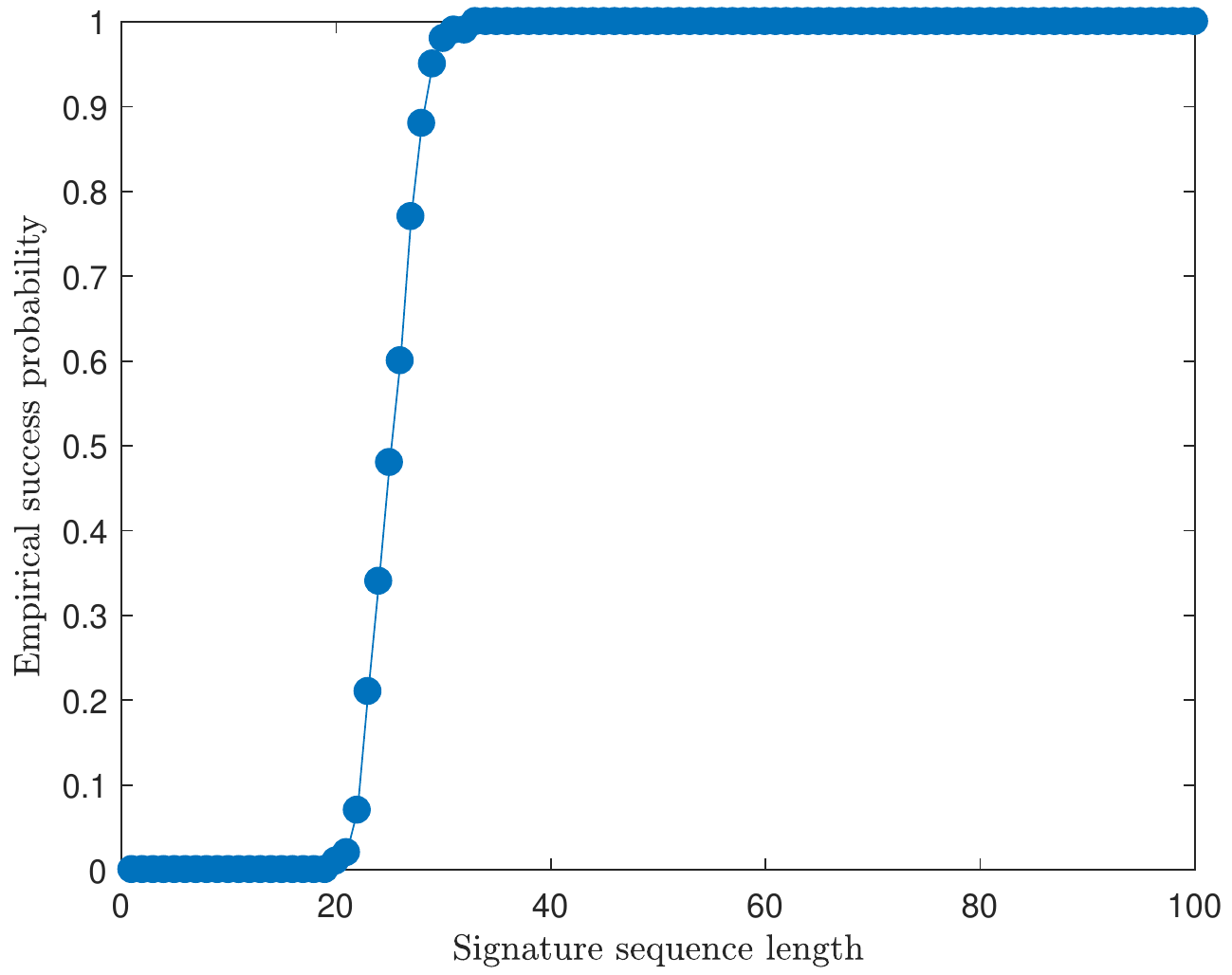}}
                %  \vspace{1.5cm}
        \end{minipage}
        \caption{Empirical success probability for solving problem $\mathscr{P}$ via CVX \cite{cvx} in the noiseless case.
                The  base station is equipped with $2$ antennas, the total number of devices is $100$ and the number of active device is $ 10 $. The channel matrix and signature matrix are generated  as  $\bm H\thicksim\mathcal{CN}(\bm 0,\bm I)$ and $\bm Q\thicksim\mathcal{CN}(\bm 0,\bm I)$, respectively. We declare successful recovery if $\norm{\bm{\hat\Theta}-\bm{\Theta_0}}_F\leq 10^{-5}$ and each point is averaged for $ 100 $ times.} \label{fig:emp_suc_prob}
\end{figure}
An example of such phenomenon is demonstrated in  Fig. \ref{fig:emp_suc_prob}, from which we see that the empirical success probability changes from $ 0 $ to $ 1 $ sharply. In particular, this indicates that when the base station is equipped with $ 2 $ antennas,  the signature sequence length around 30 is sufficient to achieve exact signal recovery for $ 100 $ devices where $ 10 $ of them are active. Thus if we can accurately find the location of the phase transition, we may choose a minimal signature sequence length accordingly to support massive IoT connectivity and channel estimation. 
  
In the following, we provide precise analysis of the location and width of the phase transition region via characterizing both success and failure conditions for signal recovery based on the conic geometry, followed by computing the probability for holding the conic optimality conditions. 
 
\subsection{Optimality Condition and Convex Geometry}
We consider the real-valued counterpart of  the statistical optimization problem $\mathscr{P}$ as follows:
%\begin{equation}
%       \begin{split}
%               \mathscr{P}_{r}:~
%               &\underset{\bm{\tilde\Theta}\in\mathbb{R}^{2N\times M}}{\mini}
%               ~\norm{\bm{\tilde Q}\bm{\tilde\Theta}-\bm{\tilde Y}}_F^2\\
%               &{\subj}~ \mathcal{R}_G(\bm{\tilde\Theta})\le \mathcal{R}_G(\bm{\tilde\Theta}_0),
%       \end{split}
%\end{equation}
\begin{equation}
\begin{split}
\mathscr{P}_{r}:~
&\underset{\bm{\tilde\Theta}\in\mathbb{R}^{2N\times M}}{\mini}
~\mathcal{R}_G(\bm{\tilde\Theta})\\
&{\subj}~ \norm{\bm{\tilde Q}\bm{\tilde\Theta}-\bm{\tilde Y}}_F \le \epsilon,
\end{split}
\end{equation}
where the linear observation in the real domain is given by
\begin{align}\label{eq:1}
        \bm{\tilde Y} &=\bm{\tilde Q}\bm{\tilde\Theta}_0+\bm{\tilde N}\notag\\
        &=\begin{bmatrix}
                \Re\left\{\bm Q\right\} &-\Im\left\{\bm Q\right\}\\
                \Im\left\{\bm Q\right\} &\Re\left\{\bm Q\right\}
        \end{bmatrix}
        \begin{bmatrix}
                \Re\left\{\bm\Theta_0\right\}\\
                \Im\left\{\bm\Theta_0\right\}
        \end{bmatrix}+
        \begin{bmatrix}
                \Re\left\{\bm N\right\}\\
                \Im\left\{\bm N\right\}
        \end{bmatrix},
\end{align}
and the regularizer is defined as 
$ \mathcal{R}_G(\bm{\tilde\Theta})=\sum_{i=1}^{N}
\norm{\bm{\tilde\Theta}_{\mathcal{V}_i}}_F $. Here
$\bm{\tilde\Theta}_{\mathcal{V}_i}=[(\bm{\tilde\theta}^i)^T,
(\bm{\tilde\theta}^{i+N})^T]^T$ is the row submatrix of $\tilde{\bm{\Theta}}$ consisting of the rows indexed by $ \mathcal{V}_i=\{i,i+N \} $. 

To facilitate phase transition analysis, problem $\mathscr{P}_{r}$ can be further approximated as the following structured group sparse estimation problem with group size $2M$:
%\begin{equation}
%       \begin{split}
%               \mathscr{P}_{\textrm{approx}} :~
%               &\underset{\bm{\tilde\Theta}\in\mathbb{R}^{2N\times M}}{\mini}
%               \quad\norm{\bm{\bar Q}\bm{\tilde\Theta}-\bm{\tilde Y}}_F^2\\
%               &{\subj}~ \mathcal{R}_G(\bm{\tilde\Theta})\le \mathcal{R}_G(\bm{\tilde\Theta}_0),
%       \end{split}
%\end{equation}
\begin{equation}
\begin{split}
\mathscr{P}_{\textrm{approx}}:~
&\underset{\bm{\tilde\Theta}\in\mathbb{R}^{2N\times M}}{\mini}
~\mathcal{R}_G(\bm{\tilde\Theta})\\
&{\subj}~ \norm{\bm{\bar  Q}\bm{\tilde\Theta}-\bm{\tilde Y}}_F \le \epsilon,
\end{split}
\label{appreal}
\end{equation}
where $\bm{\bar Q}\in\mathbb{R}^{2L\times 2N}\thicksim\mathcal{N}(\bm 0,0.5\bm I)$ is a Gaussian random matrix. The phase transition of the approximated problem $\mathscr{P}_{\textrm{approx}} $ is empirically demonstrated to coincide with the original problem $\mathscr{P}_{r}$  \cite{yang2012phase,shen2016compressed} with structured distribution in the measurement matrix $\bm{\tilde Q}$. This will be further verified in the numerical experiments in Section \ref{sec:sim}. Additionally, there are extensive empirical evidences \cite{donoho2009observed,yilmaz2017compressed}  showing that the distribution of the random measurement matrix has little effect on the locations of phase transitions. We thus focus on characterizing the phase transitions of the approximate problem $\mathscr{P}_{\textrm{approx}} $ in the real field. 

To make the presentation clear, we first characterize the phase transitions in the noiseless  case and then extend the results to the noisy case. In the noiseless  case, we rewrite problem $\mathscr{P}_{\textrm{approx}} $ as follows:
\begin{equation}\label{eq:}
        \begin{split}
                \mathscr{P}_a:~
                &\underset{\bm{\tilde\Theta}\in\mathbb{R}^{2N\times M}}{\mini}~
                \mathcal{R}_G(\bm{\tilde\Theta})\\
                &{\subj}~ \bm{\tilde Y}=\bm{\bar Q}\bm{\tilde\Theta}.
        \end{split}
\end{equation}
Problem $\mathscr{P}_a$ is said to succeed for exact recovery when it has a unique optimal points $\bm{\tilde{\Theta}}^{\ast}$, which equals the ground-truth $\bm{\tilde\Theta}_0$; otherwise, it fails. Here, the phase transition refers to the phenomenon that problem $\mathscr{P}_a$ changes from the failure state to  the successful state as the sequence length $L$ increases.
In order to establish the optimality condition for problem $\mathscr{P}_a$, we present the following definition in convex analysis \cite{amelunxen2014living}.

\begin{definition}(Descent Cone):
        The descent cone $\mathcal{D}(\mathcal{R},\bm x)$ of a proper convex function
        $\mathcal{R}:\mathbb{R}^d\rightarrow\mathbb{R}\cup\{\pm\infty\}$ at point $\bm x\in
        \mathbb{R}^d$ is the conic hull of the perturbations that do not increase
        $\mathcal{R}$ near $\bm x$, i.e.,
        \begin{equation*}
                \mathcal{D}(\mathcal{R},\bm x)=\underset{\tau>0}{\bigcup}\left\{\bm y\in\mathbb{R}^d:\mathcal{R}(\bm x+
                \tau\bm y)\leq \mathcal{R}(\bm x)\right\}.
        \end{equation*}
\end{definition}

Let ${\rm null}(\bm{\bar{Q}}, M)=\{\bm Z\in\mathbb{R}^{2N\times M}: \bm{\bar{Q}Z}=\bm{0}_{2L\times M}\}$ denote the null space of the operator $\bm{\bar{Q}}\in\mathbb{R}^{2L\times 2N}$.
With the aid of the descent cone \cite{rockafellar2015convex}, we shall establish the necessary and sufficient condition for the success of problem $\mathscr{P}_{a}$ via convex analysis \cite{chandrasekaran2012convex,amelunxen2014living}.

\begin{fact}(Optimality Condition):
        Let $\mathcal{R}$ be a proper convex function. Matrix $\bm{\tilde\Theta}_0$ is the unique optimal solution to problem $\mathscr{P}_{a}$ if and only if $\mathcal{D}(\mathcal{R}_G , \bm{\tilde\Theta}_0)\bigcap {\rm null}(\bm{\bar Q},M) = \{\bm 0\}$.
\label{fact1}
\end{fact}

\begin{figure}[htb]
        \begin{minipage}[b]{.493\linewidth}
                \centering
                \centerline{\includegraphics[width=4.45cm]{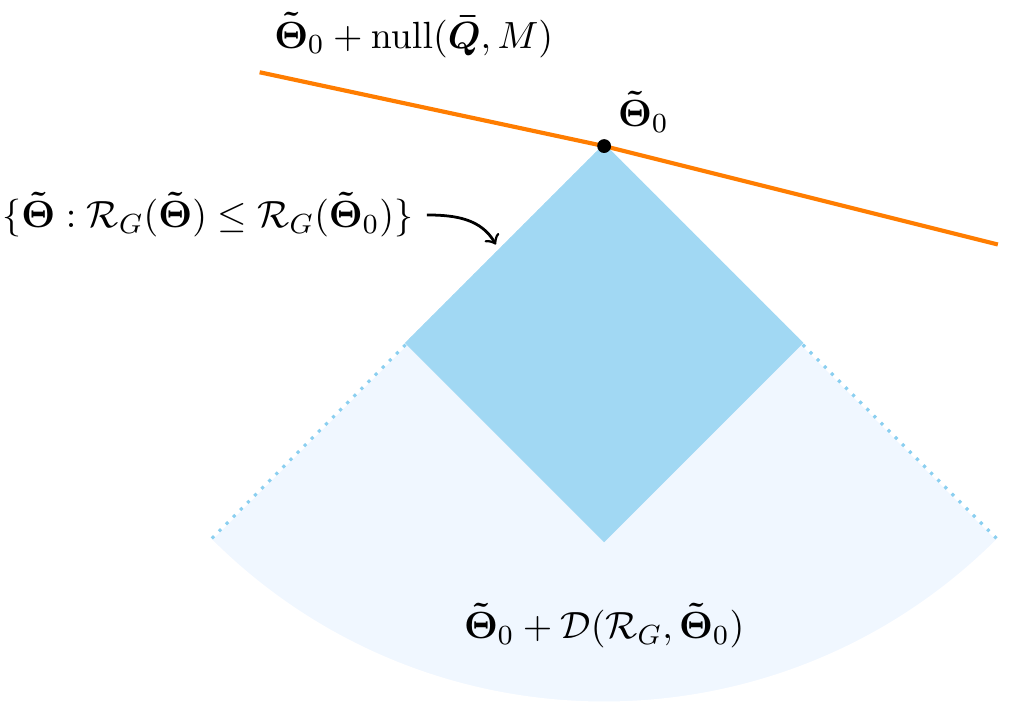}}
                %  \vspace{1.5cm}
                \centerline{(a) $\mathscr{P}_{a}$ succeeds}\medskip
        \end{minipage}
        \hfill
        \begin{minipage}[b]{0.49\linewidth}
                \centering
                \centerline{\includegraphics[width=4.45cm]{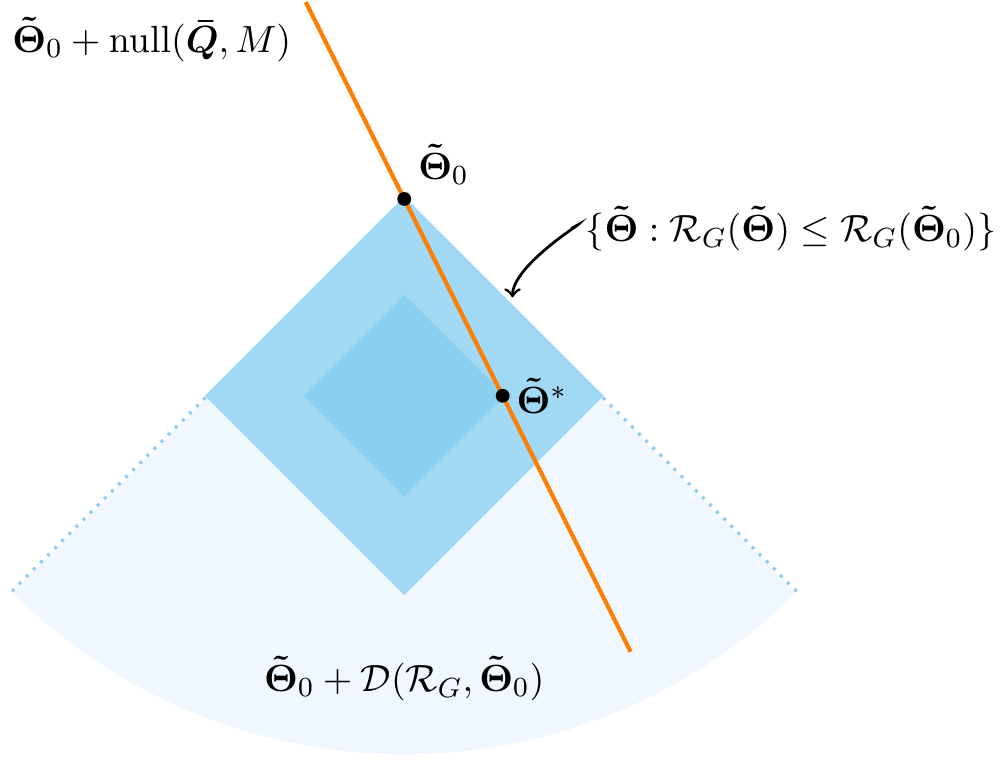}}
                %  \vspace{1.5cm}
                \centerline{(b) $\mathscr{P}_{a}$ fails}\medskip
        \end{minipage}
        \caption{Optimality condition for problem $\mathscr{P}_{a}$.}\label{fig:opt_condition}
\end{figure}
Fig. \ref{fig:opt_condition} illustrates the geometry of this optimality condition. Problem $\mathscr{P}_{a}$ succeeds if and only if the null space of $\bm{\bar Q}$ misses the cone of descent directions of $\mathcal{R}_G$ at the ground-truth $\bm{\tilde\Theta}_0$; otherwise it fails since the optimal solution is $\bm{\tilde\Theta}^{\ast}\neq\bm{\tilde\Theta}_0$ as illustrated in Fig. \ref{fig:opt_condition} (b). Intuitively, a smaller size of the decent cone  will lead to  a higher successful recovery probability of $ \mathscr{P}_a $. It is thus critical to characterize the size of the decent cone to depict the phase transition phenomena.

%\subsection{Phase Transition in Two Intersection Cones}
Based on  the optimality condition, the phase transition problem is transformed into a classic problem in conic integral geometry: {{what is the probability that a randomly rotated convex cone shares a ray with a fixed convex cone?}} The Kinematic formula \cite{schneider2008stochastic} provides an exact formula for computing this probability. However, this exact formula is hard to calculate. We thus present a practical formula that characterizes the phase transition in two intersection cones in terms of the statistical dimension \cite{amelunxen2014living}. 

\begin{definition}(Statistical Dimension):
        The statistical dimension $\delta(C)$ of a closed convex cone $C$ in $\mathbb{R}^d$ is defined as:
        \begin{equation*}
                \delta(C) = \mathbb{E}[\norm{\bm\Pi_C(\bm g)}_2^2],
        \end{equation*}
        where $\bm g\in\mathbb{R}^d$ is a standard normal vector, $\norm{\cdot}_2$ is the Euclidean norm, and $\bm\Pi_C(\bm x)= \arg\min\{\norm{\bm x-\bm y}_2:\bm
y\in C\}$ denotes the Euclidian projection onto $C$.
\end{definition}
The statistical dimension allows us to measure the size of  convex cones and is the  generalization of the dimension of linear subspaces. We state  the approximated conic kinematic formula based on the statistical dimensions of general convex cones \cite{amelunxen2014living}.
\begin{theorem}(Approximate Kinematic Formula):
        Fix a tolerance $\eta\in(0,1)$. Let $C$ and $K$ be convex cones
        in $\mathbb{R}^d$, but one of them is not a subspace. Draw a random orthogonal basis $\bm U$. Then
        \begin{equation*}
                \begin{split}
                        \delta(C)+\delta(K)\leq d-a_{\eta}\sqrt{d}&\Longrightarrow \mathbb{P}\{C\cap \bm U K\neq \{\bm 0\} \}\leq\eta \\
                        \delta(C)+\delta(K)\geq d+a_{\eta}\sqrt{d}&\Longrightarrow \mathbb{P}\{C\cap \bm U K\neq \{\bm 0\} \}\geq 1-\eta
                \end{split}
        \end{equation*}
        where $a_\eta:=\sqrt{8\log(4/\eta)}$.
        \label{thm:akf}
\end{theorem}
This theorem indicates a phase transition on whether the two randomly
rotated cones sharing a ray. That is, when the  total statistical
dimension of the two cones exceeds the ambient dimension $d$, the two randomly rotated cones share a ray with high probability; otherwise, they fail to share a ray.

\subsection{Phase Transition for Massive IoT Connectivity}
Based on general results in Theorem {\ref{thm:akf}}, we shall present the phase transition results for the exact recovery of the program $\mathscr{P}_a$ in the noiseless case and robust recovery in the noisy case. 
\subsubsection{Phase Transition in the Noiseless  Case}
To predict phase transitions of program $\mathscr{P}_a$ for signal recovery, we essentially need to compute the probability for holding the optimality condition in Fact {\ref{fact1}}. Specifically, for  Gaussian random matrix $ \bm{\bar{Q}} $, its nullity is $2N-2L  $ with probability one. Therefore, the statistical dimension of   ${\rm null}(\bm{\bar{Q}}, M)$ is $\delta({\rm null}(\bm{\bar{Q}}, M)) =  2(N-L)M $.
By replacing convex cones $C$ and $K$ in Theorem \ref{thm:akf} by the descent cone $\mathcal{D}(\mathcal{R}_G,\bm{\tilde\Theta}_0)$ and the subspace ${\rm null}(\bm{\bar{Q}}, M)$, we have the following recovery guarantees for signal recovery via program $\mathscr{P}_{a}$.
\begin{theorem}(Phase Transition of Problem $\mathscr{P}_{a}$):\label{thm:ptp3}
        Fix a tolerance $\eta\in(0,1)$. Let $\bm{\tilde\Theta}_0\in\mathbb{R}^{2N\times M}$ be a fixed matrix. Suppose $\bm{\bar Q}\in\mathbb{R}^{2L\times 2N}\thicksim\mathcal{N}(\bm 0,\bm I)$ , and let $\bm{\tilde{Y}}=\bm{\bar Q}\bm{\tilde\Theta}_0$. Then
        \begin{equation*}\small
                \begin{split}
                        2L\geq \frac{\delta(\mathcal{D}(\mathcal{R}_G,\bm{\tilde\Theta}_0))}{M} +\frac{a_{\eta}\sqrt{2NM}}{M}  &\Rightarrow \mathbb{P}\{\mathscr{P}_a ~ {\text{succeeds}} \}\geq 1-\eta \\
                        2L\leq \frac{\delta(\mathcal{D}(\mathcal{R}_G,\bm{\tilde\Theta}_0))}{M} -\frac{a_{\eta}\sqrt{2NM}}{M}  &\Rightarrow \mathbb{P}\{\mathscr{P}_a ~ \text{succeeds} \}\leq \eta
                \end{split}
        \end{equation*}
        where $a_\eta:=\sqrt{8\log(4/\eta)}$.
\end{theorem}

The above theorem indicates that $\mathscr{P}_{a}$ indeed reveals a phase transition when the signature sequence lengths $L=\frac{\delta(\mathcal{D}(\mathcal{R}_G,\bm{\tilde\Theta}_0))}{2M}$. The transition from failure to success across a sharp range with width  $\mathcal{O}(\frac{\sqrt{NM}}{M})$. The phase transition location is
thus quite accurate. We will show that the size of the decent cone of $ \mathcal{R}_G $ at a point depends solely on  its sparsity level.  

There are  mainly two implications of Theorem \ref{thm:ptp3}. First, in the absence of noise, one can see that the proposed formulation  $ \mathscr{P}_{{a}} $ allows perfect signal $ \bm{\tilde{\Theta}}_0  $ recovery with exponentially high probability if and only if the number of signature sequence length $ L $ exceeds the range of phase transition. Second,  increasing the number of antennas $ M $ in BS will narrow the range of phase transition. In particular, the width of the transition region can be reduced to zero asymptotically as the number of antennas at the BS goes to infinity.
Therefore, massive MIMO is particularly suitable for supporting massive IoT connectivity by predicting accurate phase transition location.

The sharp phase transition results are thus able to guide the selection of the signature sequence length. We will further contribute this work by computing the statistical dimension of the descent cone  $\mathcal{D}(\mathcal{R}_G,\bm{\tilde\Theta}_0)$ for the group sparse inducing norm in Section {\ref{stadim}}. 

\subsubsection{Phase Transition in the Noisy Case}
\label{subsec:noise}
Let $ \bm{\tilde\Theta}^{\ast} $ be an estimate of the ground truth matrix $ \bm{\tilde{\Theta}}_0 $. 
To evaluate the accuracy of the estimator, we define the average squared prediction error as follows:
\begin{equation}\label{eq:ase}
        R(\bm{\tilde\Theta}^{\ast})=\frac{1}{2LM}\norm{\bm{\bar Q}\bm{\tilde\Theta}^{\ast}-\bm{\bar Q\tilde\Theta}_0}_F^2.
\end{equation}
We further define the estimation error of the estimator as $ \mathbb{E}_{\bm{\tilde{N}} }[R(\bm{\tilde\Theta}^{\ast})] $ for a given signature matrix $ \bm{\bar{Q}} $ and ground truth matrix $\bm{\tilde{\Theta}}_0 $.
We will see this quantity enjoys a phase transition as $L$ varies.

%Oymak and Hassibi \cite{oymak2016sharp} studied the stability of this phase transition in the presence of noise. They considers the formulation (with $ M=1 $)
To facilitate efficient analysis in the noisy case, we consider the following formulation:
\begin{equation}
        \begin{split}
                \mathscr{P}_{b} :~
                &\underset{\bm{\tilde\Theta}\in\mathbb{R}^{2N\times M}}{\mini}
                \quad\norm{\bm{\bar Q}\bm{\tilde\Theta}-\bm{\tilde Y}}_F^2\\
                &{\subj}~ \mathcal{R}_G(\bm{\tilde\Theta})\le \mathcal{R}_G(\bm{\tilde\Theta}_0),
        \end{split}
\end{equation}
which is equivalent to problem $ \mathscr{P}_{\textrm{approx}} $ for some choice of the parameter $\epsilon$. It turns out that this problem also undergoes a phase transition when the length of the signature sequence is picked as $L=\frac{\delta(\mathcal{D}(\mathcal{R}_G,\bm{\tilde\Theta}_0))}{2M}$, which is coincident with the noiseless case \cite{oymak2016sharp}. We shall provide sharp phase transition results for robust group sparse estimation via program $\mathscr{P}_b$ in the following theorem.  

%They  \cite{oymak2016sharp} proved that, for $ M=1 $, problem $\mathscr{P}_{\textrm{approx}}$ reveals a phase transition  when $\bm{\bar{Q}\bar{Q}^T=\bm I}$, or $\bm{\bar{Q}}$ has independent standard Gaussian entries. We generalize their results for $ M\ge 1 $ gives the following theorem.
\begin{theorem}(Phase Transition of Problem $\mathscr{P}_{\textrm{b}}$):\label{thm:robust_recovery}
        Assume matrix $\bm{\bar{Q}}\in\mathbb{R}^{2L\times 2N}$ satisfies $\bm{\bar{Q}\bar{Q}^T=\bm I}$. Let the noise matrix $\bm{\tilde N}\thicksim\mathcal{N}(\bm 0,\sigma^2\bm I)$ be independent of $\bm{\bar{Q}}$ and $\bm Y=\bm{\bar Q}\bm{\tilde\Theta}_0+\bm{\tilde N}$
        with $\bm{\tilde\Theta}_0\in\mathbb{R}^{2N\times M}$. Let $\bm{\tilde\Theta}^{\ast}$ denote the optimal solution to problem $\mathscr{P}_b$. The prediction error $R(\bm{\tilde\Theta}^{\ast})$ and empirical error $ \hat R(\bm{\tilde\Theta}^{\ast}) $ is defined as $R(\bm{\tilde\Theta}^{\ast})=
        \frac{1}{2LM}\norm{\bm{\bar{Q}}\bm{\tilde\Theta}^{\ast}-
                \bm{\bar{Q}}\bm{\tilde\Theta}_0}_F^2$, $\hat R(\bm{\tilde\Theta}^{\ast})=
        \frac{1}{2LM}\norm{\bm{\bar{Q}}\bm{\tilde\Theta}^{\ast}-\bm Y}_F^2$, respectively. Set $\delta= \frac{\delta(\mathcal{D}(\mathcal{R}_G,\bm{\tilde\Theta}_0))}{2M}$. Then there exist constants $c_1,c_2>0$ such that
        \begin{itemize}
                \item Whenever $L<\delta$,
                \begin{align}
                        \underset{\sigma>0}{\max}~\frac{\mathbb{E}_{\bm{\tilde{N}}}[R(\bm{\tilde\Theta}^{\ast})]}{\sigma^2}&=1,\label{eq:noise1} \\
                        \underset{\sigma\rightarrow 0}{\lim}\frac{\mathbb{E}_{\bm{\tilde{N}}}[\hat{R}(\bm{\tilde{\Theta}}^\ast)]}{\sigma^2}&=0, \label{eq:noise11}
                \end{align}
                with probability $1-c_1\exp(-c_2(L-\delta)^2/(NM^3))$.
                \item Whenever $L>\delta$,
                \begin{align}
                        \left|\underset{\sigma>0}{\max}~\frac{\mathbb{E}_{\bm{\tilde{N}}}[R(\bm{\tilde\Theta}^{\ast})]}{\sigma^2}-\frac{\delta}{L}\right|&\leq\frac{t\sqrt{2NM}}{2LM},\label{eq:noise2}\\
                        \left|\underset{\sigma\rightarrow 0}{\lim}~\frac{\mathbb{E}_{\bm{\tilde{N}}}[\hat{R}(\bm{\tilde{\Theta}}^\ast)]}{\sigma^2}-\left(1-\frac{\delta}{L}\right)\right|&\leq\frac{t\sqrt{2NM}}{2LM},\label{eq:noise22}
                \end{align}
                with probability $1-c_1\exp(-c_2t^2)$.
        \end{itemize}
        Here, the probabilities are calculated over the random measurement matrix $\bm{\bar Q}$.
\end{theorem}
\begin{proof}
        Please refer to Appendix \ref{app:thm_robust_recovery} for details.
\end{proof}

This theorem describes a phase transition at location $\delta$ in the noisy case, which extends the results in the noiseless case. When the  signature sequence length $L$ is smaller than $\delta$, the worst-case  estimation error $ \mathbb{E}_{\bm{\tilde{N}} }[R(\bm{\tilde\Theta}^{\ast})] $ is simply the noise power $\sigma^2$, and increasing $L$ cannot decrease the estimation error. This means that the regularized linear regression problem is sensitive to noise. After crossing the phase transition, increasing the signature length can reduce the worst-case estimation error at the rate $1/L$. The worst-case estimation error is achieved when $\sigma\rightarrow 0$ \cite{oymak2016sharp}. It will be verified in section \ref{sec:sim} that the obtained phase transition results accurately depict the phase transition behavior of the original problem $\mathscr{P}$. One  observation in Theorem \ref{thm:robust_recovery} is that the behavior of empirical estimation error $ \hat R(\bm{\tilde\Theta}^{\ast}) $ provides guidance for choosing parameter $ \epsilon $ in problem $ \mathscr{P}_{\textrm{approx}} $.  Using the worst case empirical estimation error, we can set
\begin{equation}\label{eq:eps}
\epsilon=\sigma\sqrt{2LM-\delta(\mathcal{D}(\tilde{\mathcal{R}}_G,\bm{\tilde\Theta}_0))},
\end{equation}
provided a reasonable estimate of noise power $ \sigma^2 $.    

\subsection{Computing the Statistical Dimension}
\label{stadim}
Theorem \ref{thm:ptp3} and Theorem \ref{thm:robust_recovery} allow us to sharply locate the phase transitions for $\mathscr{P}_{a}$ and $\mathscr{P}_{b}$, respectively, and computing the statistical dimension of the descent cone is the key to evaluate  the theoretical results. But this presents its own challenges to provide  a computationally feasible formula for the statistical dimension. We thus provide an accurate estimate and insightful expression for $\delta(\mathcal{D}(\mathcal{R}_G,\bm{\tilde\Theta}_0))$ using the following  recipe  suggested in \cite{amelunxen2014living}.
\begin{lemma}(The Statistical Dimension of a Descent Cone):
\label{lemmasta}
        Let $\mathcal{R}:\mathbb{R}^d\rightarrow\mathbb{R}\cup\{\pm\infty\}$ be a proper convex function and $\bm x\in\mathbb{R}^d$. Assume that the sub-differential $\partial\mathcal{R}(\bm x)$ is non-empty, compact, and does not contain the origin. Then
        \begin{equation}\label{eq:stdim_upbnd}
                \delta(\mathcal{D}(\mathcal{R},\bm x))\leq \underset{\tau\geq 0}{\rm inf}~ \mathbb{E}[{\rm dist}^2(\bm g, \tau\cdot\partial\mathcal{R}(\bm x))],
        \end{equation}
        where $\bm g\in\mathbb{R}^n$ is a standard normal vector.
\end{lemma}

Although Lemma {\ref{lemmasta}} suggested a general method to study the statistical dimension of a descent cone, it still needs additional technical effort to compute accurate estimate for the statistical dimension of a descent cone for the group sparsity inducing norm adopted in this paper. 

\begin{proposition}(Statistical Dimension  for $\mathcal{R}_G$):\label{pro:1}
        Let $\bm\Theta_0\in\mathbb{C}^{N\times M}$ be with $S$ nonzero rows, and define the
        normalized sparsity $\rho:=S/N$. The upper bound of  statistical dimension
        of descent cone of $\mathcal{R}_G$ at $\bm{\tilde\Theta_0}=[(\Re \{\bm\Theta_0\})^T,(\Im \{\bm\Theta_0)^T\}]^T\in\mathbb{R}^{2N\times M}$
        is given by
        \begin{equation}\label{eq:stdim_Rg}
                \begin{split}
                        &\frac{\delta(\mathcal{D}(\mathcal{R}_G;\bm{\tilde\Theta_0}))}{N}\leq
                        \underset{\tau\geq 0}{\inf}\Big\{\rho(2M+\tau^2)\\
                        &+(1-\rho)\frac{2^{1-M}}
                        {\Gamma(M)}\int_{\tau}^{\infty}    (u-\tau)^2u^{2M-1}e^{-\frac{u^2}{2}}\mathrm{d}u\Big\}.
                \end{split}
        \end{equation}
        The unique optimum $\tau^\star$ which minimizes the right-hand side
        of \eqref{eq:stdim_Rg} is the solution of
        \begin{equation}
                \frac{2^{1-M}}{\Gamma(M)}\int_{\tau}^{\infty}(\frac{u}{\tau}-1)u^{2M-1}e^{-\frac{u^2}{2}}\mathrm{d}u
                =\frac{\rho}{1-\rho}.
        \end{equation}
\end{proposition}
\begin{proof}
        Please refer to Appendix \ref{app:proof_pro1} for details.
\end{proof}

The bound provided in Proposition \ref{pro:1} can be numerically computed efficiently, and thus can be utilized in Theorem \ref{thm:ptp3} and Theorem \ref{thm:robust_recovery} to compute the locations of phase transitions. Note that the bound only depends  on the sparsity level of matrix $\bm{\tilde{\Theta}}_0 $ and turns out to be accurate via extensive experiments.

\section{Sharp Computation and Estimation Trade-offs via Smoothing Method}
\label{sec:tradeoff}
In an IoT network with a  massive number of devices, it becomes  critical to solve the JADE problem under a fixed time budget. To address the computational challenges for solving the high-dimensional group sparsity estimation problem, we adopt the smoothing method to smooth the non-differentiable group sparsity inducing regularizer to accelerate the convergence rates. We further characterize the sharp trade-off between the computational cost and estimation accuracy. This provides guidelines on choosing the optimal signature sequences to maintain the estimation accuracy for the smoothed group sparsity estimator.   

\subsection{Accelerating Convergence Rate via Smoothing}
Adding a smooth function to ``smooth" the non-differentiable objective function is a well-known idea in the context of sparse optimization, which makes the regularized problem easy to solve \cite{lai2013augmented, bruer2015designing}.  In particular, for problem $\mathscr{P}$, we augment $\mathcal{R}(\bm{\Theta})$ by adding a smoothing function $\frac{\mu}{2}\norm{\bm \Theta}_F^2$, where $\mu$ is a positive scalar and called as the smoothing parameter. Problem $\mathscr{P}$ is thus smoothed as\begin{equation}
\begin{split}
\mathscr{P}_{s}:\quad
&\underset{\bm{\Theta}\in\mathbb{C}^{N\times M}}{\mini}
\quad\tilde{\mathcal{R}}(\bm{\Theta}):=\mathcal{R}(\bm{\Theta})+ \frac{\mu}{2}\norm{\bm
\Theta}_F^2\\
&{\subj}\quad \norm{\bm{ Q\Theta}-\bm{ Y}}_F\leq \epsilon,
\end{split}
\end{equation}
which can be rewritten in the real domain as follows,
 
\begin{equation}
        \begin{split}
                \mathscr{P}_{\tilde{r}}:\quad
                &\underset{\bm{\tilde\Theta}\in\mathbb{R}^{2N\times M}}{\mini}
                \quad\tilde{\mathcal{R}}_G(\bm{\tilde\Theta})\\
                &{\subj}\quad \norm{\bm{\tilde Q\tilde\Theta}-\bm{\tilde
Y}}_F\leq \epsilon,
        \end{split}
\end{equation}
where $\bm{\tilde\Theta}$, $\bm{\tilde
Q}$, and $\bm{\tilde
Y}$ are given in problem $\mathscr{P}_{\textrm{approx}}$ (\ref{appreal}).

As problem $ \mathscr{P}_{\tilde{r}} $ is not differentiable, applying the subgradient method to solve it  would yield a slow coverage rate. Fortunately, the dual formulation of problem $ \mathscr{P}_{\tilde{r}} $ leverages the benefits from smoothing techniques, as the smoothed dual problem can be reduced to an unconstrained problem with the composite objective function consisting of a convex, smooth function and a convex, nonsmooth function. This composite form  can be solved by a rich set of  first-order methods  such as Auslender and Teboulle’s algorithm  \cite{auslender2006interior}, Nesterov’s 2007 algorithm (N07) \cite{nesterov2013gradient} and Lan, Lu, and Monteiro’s modification of N07 (LLM) algorithm \cite{lan2011primal} etc., and these algorithms have the $ O(1/\sqrt{\gamma}) $ ($ \gamma $ is the numerical accuracy) convergence rate \cite{nemirovskii1983problem, becker2011templates}. 

The dual problem of $\mathscr{P}_{\tilde{r}}$ is given by 
\begin{equation*}
        \begin{split}
                &\underset{\bm{Z},t }{\maxi}
                \quad\mathcal{D}(\bm Z,t):=\underset{\bm{\tilde\Theta}}{\inf}\left\{  \tilde{\mathcal{R}}(\bm{\tilde\Theta})-\langle\bm Z, \bm{\tilde Q\tilde\Theta}-\bm{\tilde Y}\rangle-t\epsilon  \right\}\\
                &{\subj}\quad \norm{\bm Z}_F\leq t,
        \end{split}
\end{equation*}
where $\bm {Z}\in\mathbb{R}^{2N\times M}$ and $t>0$. 
Since $\epsilon\ge0$,  eliminating the dual variable $t$, we obtain the unconstrained
problem as follows
\begin{equation*}
        \underset{\bm Z\in\mathbb{R}^{2N\times M}}{\mini}
        ~\mathcal{D}(\bm Z):=-\underset{\bm{\tilde\Theta}}{\inf}\left\{  \tilde{\mathcal{R}}(\bm{\tilde\Theta})-\langle\bm Z, \bm{\tilde Q\tilde\Theta}-\bm{\tilde Y}\rangle-\epsilon\norm{\bm Z}_F  \right\}.
\end{equation*}
The dual objective function can be further represented as the following composite function
\begin{align}\label{eq:comp}
        \mathcal{D}(\bm Z)=\underbrace{-\underset{\bm{\tilde\Theta}}
                {\inf}\left\{\tilde{\mathcal{R}}(\bm{\tilde\Theta})-\langle\bm Z, \bm{\tilde Q\tilde\Theta}\rangle\right\}-\langle\bm Z, \bm{\tilde Y}\rangle}_{\tilde{\mathcal{D}}(\bm Z)}
        \underbrace{\underset{}{+\epsilon\norm{\bm Z}_F}}_{\mathcal{H}(\bm Z)}.
\end{align}
Function $\tilde{\mathcal{D}}(\bm Z)$ is differentiable and its gradient is
\begin{equation*}
        \nabla \tilde{\mathcal{D}}(\bm Z) = -\tilde{\bm Y}+\tilde{\bm Q}\tilde{\bm\Theta}_{\bm Z},
\end{equation*}
where
\begin{equation}\label{eq:theta_Z}
        \tilde{\bm\Theta}_{\bm Z}:=\underset{\bm{\tilde \Theta}}{\arg\min}\left\{\tilde{\mathcal{R}}(\bm{\tilde\Theta})-\langle\bm Z, \bm{\tilde Q\tilde\Theta}\rangle\right\}.
\end{equation}
Furthermore, $\nabla \tilde{\mathcal{D}}(\bm Z)$ is a Lipschitz continuous with Lipschitz constant upper bounded by $L_s := \mu^{-1}\norm{\bm{\tilde{Q}}}_2^2$. That is to say, the dual objective is a composition of the smooth function $\tilde{\mathcal{D}}(\bm Z)$ and the nonsmooth function $\mathcal{H}(\bm Z)$. This composite form \eqref{eq:comp} can be solved by a rich set of first-order methods \cite{becker2011templates}, which are particularly sensitive to the smoothing parameter $ \mu $, i.e., a larger value of the smoothing parameter $ \mu $  leads to a faster convergence rate.
 
In particular, we present the  Lan, Lu, and Monteiro’'s algorithm \cite{lan2011primal}  in Algorithm \ref{algo} as a typical example to show the benefits of smoothing. 

\begin{algorithm}[h]
        \SetKwData{Left}{left}\SetKwData{This}{this}\SetKwData{Up}{up}
        \SetKwInOut{Input}{Input}\SetKwInOut{Output}{output}
        \Input{Signature matrix $\bm{\tilde{Q}}\in\mathbb{R}^{2L\times2N}$, observed matrix $\bm{\tilde{Y}}\in\mathbb{R}^{2L\times M}$, and parameter $\epsilon$.}
        
        $\bm Z_0\leftarrow \bm 0$, $\bm{\bar{Z}}_0\leftarrow \bm Z_0$, $t_0\leftarrow 1$\\
        
        \For{$k=0,1,2,\cdots$}{
                $\bm B_k\leftarrow (1-t_k)\bm Z_k+t_k\bm{\bar Z}_k$ \\
                
                $\bm{\tilde{\Theta}}_k\leftarrow\mu^{-1}{\rm SoftThreshold}(\bm{\tilde{Q}}^{\rm T}\bm B_k,1) $\nllabel{alg:line0}\\
                
                $\bm{\bar{Z}}_{k+1}\leftarrow{\rm Shrink}(\bm{\bar{Z}}_k-(\bm{\tilde{Q}}\bm{\tilde{\Theta}}_k-\bm{\tilde{Y}})/L_s/t_k,\epsilon/L_s/t_k)$\nllabel{alg:line1}\\
                
                $\bm Z_{k+1}\leftarrow{\rm Shrink}(\bm B_k-(\bm{\tilde{Q}}\bm{\tilde{\Theta}}_k-\bm{\tilde{Y}})/L_s,\epsilon/t_k)$ \nllabel{alg:line2}\\
                
                $t_{k+1}\leftarrow 2/(1+(1+4/t_k^2)^{1/2})$
        }
        \caption{Lan, Lu, and Monteiro's Algorithm}\label{algo}
\end{algorithm}\DecMargin{1em}

In Algorithm \ref{algo}, lines \ref{alg:line0} is the solution to \eqref{eq:theta_Z}, line \ref{alg:line1} and \ref{alg:line2} are the solutions to the following composite gradient mapping  respectively,
\begin{equation*}\small
        \begin{split}
                \bm{\bar{Z}}_{k+1} &\leftarrow \underset{\bm Z\in\mathbb{R}^{2N\times M}}{\arg\min}\left\{\langle \nabla \tilde{\mathcal{D}}(\bm Z), \bm Z\rangle  +\frac{1}{2}t_kL_s\norm{\bm Z-\bm{\bar{Z}}_k}_F+\mathcal{H}(\bm Z)\right\}, \\
                \bm{Z}_{k+1} &\leftarrow \underset{\bm Z\in\mathbb{R}^{2N\times M}}{\arg\min}\left\{\langle \nabla \tilde{\mathcal{D}}(\bm Z), \bm Z\rangle  +\frac{1}{2}L_s\norm{\bm Z-\bm{B}_k}_F+\mathcal{H}(\bm Z)\right\}.
        \end{split}
\end{equation*}
The operator  ${\rm Shrink}$ is given by
\begin{equation*}
        {\rm Shrink}(\bm Z,t) = \max\left\{1-\frac{t}{\norm{\bm Z}_F},0 \right\} \bm Z.
\end{equation*}
Let $\bm X={\rm SoftThreshold}(\bm Z,t)\in\mathbb{R}^{N\times M}$. Each row of $\bm X$ is given by
\begin{equation*}
        \bm x^i={\rm Shrink}(\bm z^i,t),\quad\text{for}~ i=1,\cdots N.
\end{equation*}

Let $ \bm Z^\ast $ be  an optimal point for \eqref{eq:comp}, then the convergence behavior of Algorithm \ref{algo} satisfies \cite{becker2011templates},
\begin{equation}\label{key}
        \mathcal{D}(\bm Z_{k+1})-\mathcal{D}(\bm Z^\ast)\le \frac{2 \norm{\bm{\tilde{Q}}}_2^2 \|\bm Z_0-\bm Z^\ast\|_F^2 }{\mu k^2}.
\end{equation}
Therefore, the number of iterations required to reach $ \gamma $ accuracy is at most $\left\lceil \sqrt{\frac{2\norm{\bm{\tilde{Q}}}_2^2}{\mu\gamma }}\|\bm Z_0-\bm Z^\ast\|_F \right\rceil$, which implies that a larger $ \mu $ will result in a faster convergence rate. For each iteration in Algorithm \ref{algo}, the operators $ \operatorname{SoftThreshold} $ and $ \operatorname{Shrink} $ are computationally cheap, and the dominate cost is the matrix-matrix products involving the signature matrix $ \bm{\tilde{Q}} $, which is $ \mathcal{O}(LNM) $. 

In practice, we terminate the algorithm   when the relative primal feasibility gap satisfies $ |\norm{\bm{\tilde Q\tilde\Theta}_k-\bm{\tilde Y}}_F- \epsilon|/\epsilon\le\gamma_0$ for a small enough $ \gamma_0 $. 
The bound of the feasibility gap of primal iterates $ \bm{\tilde\Theta}_k  $ at each iteration $ k $ is given as follows \cite{bruer2015designing},
\begin{equation}\label{key}
        \left|\norm{\bm{\tilde Q\tilde\Theta_k}-\bm{\tilde Y}}_F-
        \epsilon\right|\le\frac{2\norm{\bm{\tilde{Q}}}_2^2\norm{\bm Z^{\ast}}_F }{\mu k}.
\end{equation}
Therefore, the number of iterations sufficient for convergence is upper bounded as
\begin{equation}\label{key}
        k\le\frac{2\norm{\bm{\tilde{Q}}}_2^2\norm{\bm Z^{\ast}}_F}{\gamma_0\mu\sigma\sqrt{2LM-\delta(\mathcal{D}(\tilde{\mathcal{R}}_G,\bm{\tilde\Theta}_0))}},
\end{equation}
which shows the number of iterations required for convergence in terms of the smoothing parameter, signature sequence length and solution accuracy. We will show in Fig. {\ref{fig:convergence_vs_mu}} that the convergence rate of the smoothed estimator $\mathscr{P}_s$ will be accelerated as  the smoothing parameter increases.

\subsection{Computation and Estimation Trade-offs}
From the geometric perspective, the smoothing term in $ \tilde{\mathcal{R}}(\bm{\Theta}) $ (with $ \mu>0 $)  enlarges the sublevel set of the regularizer $ \mathcal{R}(\bm{\Theta}) $, which results in a problem that is computationally easier to solve  with  an accelerated convergence rate. However, this geometric deformation brings a loss in the estimation accuracy according to the phase transition results in Theorem \ref{thm:robust_recovery}.  This results in a trade-off between the computational time and estimation accuracy.   
The trade-off is controllable given the statistical dimension of the decent cone of the smoothed regularizer $\tilde{\mathcal{R}}_G(\bm{\tilde\Theta})=\mathcal{R}_G(\bm{\tilde{\Theta}})+\frac{\mu}{2}\norm{\bm{\tilde \Theta}}_F^2$. In particular, the statistical dimension $\delta(\mathcal{D}(\tilde{\mathcal{R}}_G,\bm{\tilde\Theta}_0))$ can be accurately estimated by the following result.

\begin{proposition}(Statistical Dimension Bound for $\tilde{\mathcal{R}}_G$ )\label{pro:2}
        Let $\bm\Theta_0\in\mathbb{C}^{N\times M}$ be with $S$ nonzero rows, and define the
        normalized sparsity as $\rho:=S/N$. An upper bound of  the statistical dimension
        of the descent cone of $\tilde{\mathcal{R}}_G$ at $\bm{\tilde\Theta_0}=[(\Re \{\bm\Theta_0\})^T,(\Im \{\bm\Theta_0)^T\}]^T\in\mathbb{R}^{2N\times M}$
        is given by
        \begin{equation}\label{eq:stadim3}
                \begin{split}
                        &\frac{\delta(\mathcal{D}(\tilde{\mathcal{R}}_G;\bm{\tilde\Theta}_0))}{N}\leq
                        \underset{\tau\geq0}{\inf}\Big\{\rho(2M+\tau^2(1+2\mu\bar{a}+\mu^2\bar{b}))\\           
                        &+ (1-\rho)\frac{2^{1-M}}{\Gamma(M)}\int_{\tau}^{\infty}
                        (u-\tau)^2u^{2M-1}e^{-\frac{u^2}{2}}\mathrm{d} u\Big\}.
                \end{split}
        \end{equation}
        The unique optimum $\tau^\star$ which minimizes the right-hand side of \eqref{eq:stadim3} is the solution of
        \begin{equation}
                \begin{split}
                        \frac{2^{1-M}}{\Gamma(M)}\int_{\tau}^{\infty}\left(\frac{u}{\tau}-1\right)u^{2M-1}e^
                        {-\frac{u^2}{2}}\mathrm{d}u
                        =\frac{\rho(1+2\mu\bar{a}+\mu^2\bar{b})}{1-\rho},
                \end{split}
        \end{equation}
        where $\bar{a}=\frac{1}{S}\sum_{i=1}^{S}\norm{(\bm{\tilde\Theta}_0)_{\mathcal{V}_i}}_F$, $\bar{b}=\frac{1}{S}\sum_{i=1}^{S}\norm{(\bm{\tilde\Theta}_0)_{\mathcal{V}_i}}_F^2$.
\end{proposition}
\begin{proof}
        Please refer to Appendix \ref{app:proof_pro2} for details.
\end{proof}
Note that $\bar{a}$ and $\bar{b}$ can be calculated given the distribution of the ground truth $\bm{\Theta}_0$. For instance, with $\bm{\Theta}_0\thicksim\mathcal{CN}(\bm 0,2\sigma^2\bm I)$, we have $\bm{\tilde\Theta}_0\thicksim\mathcal{N}(\bm 0,\sigma^2\bm I)$. Here,
$\norm{(\bm{\tilde\Theta}_0)_{\mathcal{V}_i}}_F/\sigma$ follows chi distribution
with $2M$ degrees of freedom and $\norm{(\bm{\tilde\Theta}_0)_{\mathcal{V}_i}}^2_F/\sigma^2$ follows chi square
distribution with $2M$ degrees of freedom. Hence, we can set $\bar{a}=\sqrt{2}\frac{\Gamma((2M+1)/2)}{\Gamma(M)}\sigma$, $\bar{b}=2M\sigma^2$. 

Although the convergence rate can be accelerated by increasing the smoothing
parameter as shown in the previous subsection, Proposition \ref{pro:2} suggests that a larger smoothing parameter
results in a larger statistical dimension $ \delta(\mathcal{D}(\tilde{\mathcal{R}}_G,\bm{\tilde\Theta}_0))$ as the bound in \eqref{eq:stadim3} grows with $\mu$.
This will reduce the estimation accuracy  for a given signature sequence length according to the result in Theorem
\ref{thm:robust_recovery}. Fig. {\ref{fig:estimation_error}} will demonstrate that the estimation error indeed will increase as the smoothing parameter becomes large.  Therefore, the smoothing method yields a trade-off between the computational cost and estimation accuracy, as increasing the smoothing parameter will improve the convergence rate while reduce the estimation accuracy. Such a tradeoff is particular important in scenarios
with massive IoT devices and a limited time budget, but not very stringent requirement on estimation accuracy. 

\subsection{Discussion}
For typical IoT applications, we are particularly interested in reducing the overall computational cost while maintaining the estimation accuracy, which can be achieved by interpreting the above trade-off from another perspective.  For the smoothed estimator $\mathscr{P}_s$, Proposition \ref{pro:2} together with Theorem \ref{thm:robust_recovery} can help to provide guidelines for choosing a minimal signature sequence length to maintain the estimation accuracy  for a given smoothing parameter $ \mu $. Specifically, while smoothing may increase the estimation error, we  can increase the signature sequence for the smoothed estimator $\mathscr{P}_s$ compared with the original nonsmooth estimator $\mathscr P$. Specifically, given a smoothing parameter $ \mu $, 
according to Theorem \ref{thm:robust_recovery}, we are able to maintain the
estimation accuracy by choosing the  signature sequence length $ L $ as follows
\begin{equation}\label{eq:alloc_sl}
L = \frac{\delta(\mathcal{D}(\tilde{\mathcal{R}}_G(\mu),\bm{\tilde\Theta}_0))}{2M\gamma_1},
\end{equation}
where $ \gamma_1 = \underset{\sigma>0}{\max}~\frac{\mathbb{E}_{\bm{\tilde{N}}}[R(\bm{\tilde\Theta}^{\ast})]}{\sigma^2} $  is the expectation of the worst-case estimation accuracy normalized by noise power $ \sigma^2 $.

\section{Simulation results}
\label{simu}
In this section, we   verify the phase transition phenomena in IoT networks characterized by Theorem \ref{thm:ptp3} and Theorem \ref{thm:robust_recovery} via simulations. We further simulate the developed dual-smoothed algorithm to illustrate the benefits of smoothing, as well as the trade-offs between the estimation accuracy and computational cost.
\begin{figure}[t]
                \begin{minipage}[b]{1\linewidth}
                        \centering
                        \centerline{\includegraphics[width=0.95\columnwidth]{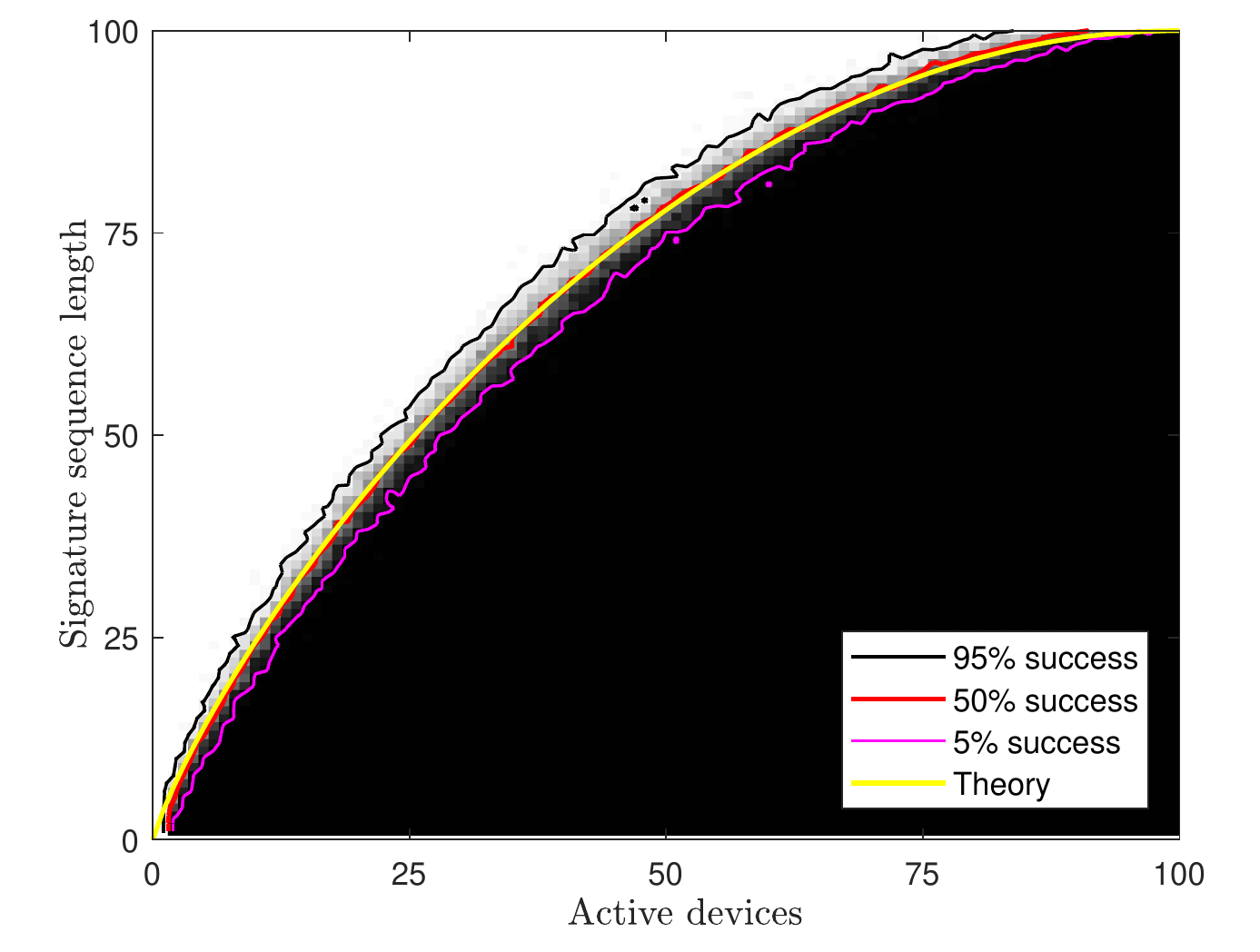}}
                        %\vspace{1.5cm}
                        \centerline{\small(a) The noiseless  case}\medskip
                \end{minipage}
                \hfill
                \begin{minipage}[b]{1\linewidth}
                        \centering
                        \centerline{\includegraphics[width=0.95\columnwidth]{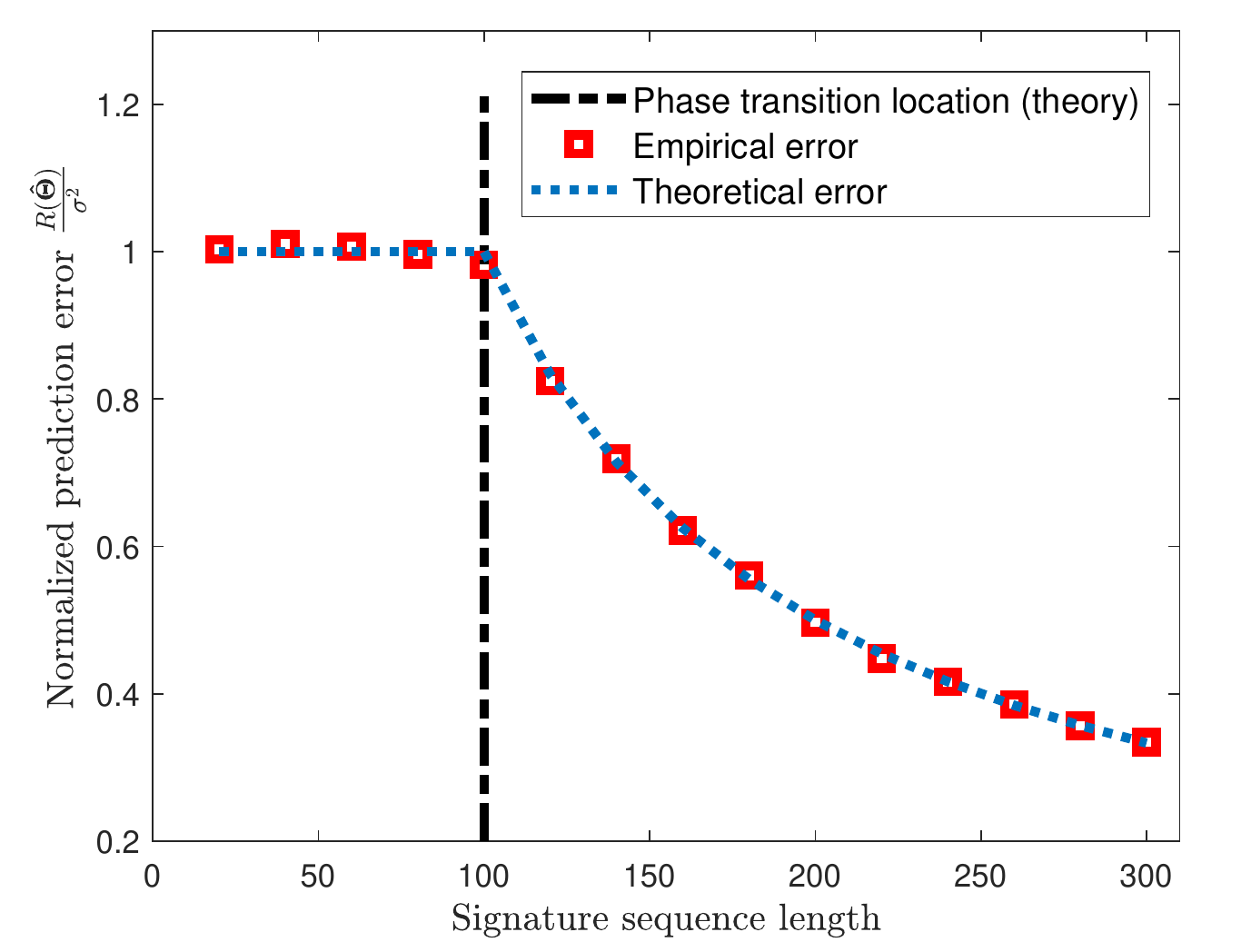}}
                        %\vspace{1.5cm}
                        \centerline{\small(b) The noisy case}\medskip
                \end{minipage}
      
        \caption{Phase transitions in massive device connectivity. }
        \label{fig:res}
\end{figure}
 
\label{sec:sim}
\subsection{Phase Transitions}
To verify the phase transition in the noiseless  case, we consider the scenario
in which the  base station is equipped with $2$ antennas, and the total number of devices is $100$. For estimation problem $\mathscr{P}$ in this noiseless setting, the channel matrix and signature matrix are generated  as  $\bm H\thicksim\mathcal{CN}(\bm 0,\bm I)$ and $\bm Q\thicksim\mathcal{CN}(\bm 0,\bm I)$, respectively. We declare successful recovery if $\norm{\bm{\hat\Theta}-\bm{\Theta_0}}_F\leq 10^{-5}$, and we record the success probability from $50$ trials. The experiments are performed using the CVX package \cite{cvx} in Matlab with default settings.

In Fig. \ref{fig:res} (a), we show the probability of successful recovery as a function of the signature sequence length and the number of active devices. The brightness corresponds to the  empirical recovery probability (white = 100\%, black = 0\% ). On top of this heap map, the empirical curves of $5\%$, $50\%$, $95\%$  are success probabilities that are calculated from data. It can be seen that the theoretical curve from Theorem \ref{thm:ptp3} closely matches the empirical curve of the $50\%$ success probability.

To verify the phase transition in the noisy case, we consider a scenario where the base station is equipped with $3$ antennas, and the total number of devices is  $300$. We fix the number of active devices as $|\mathcal{S}|=42$, hence the theoretical phase transition location is given as $\frac{\delta(\mathcal{D}(\mathcal{R}_G,\bm{\tilde\Theta}_0))}{2M}\approx100$.
For estimation problem $\mathscr{P}$, the channel matrix is generated as
 $\bm H\thicksim\mathcal{CN}(\bm 0,\bm I)$, the signature matrix as $\bm Q\thicksim\mathcal{CN}(\bm 0,\bm I)$ and the additive noise matrix as $\bm N\thicksim\mathcal{CN}(\bm 0,0.001\bm I)$. The simulation results are  averaged for $100$ times.

In Fig. \ref{fig:res} (b), we see that the normalized prediction error can be accurately predicted by Theorem \ref{thm:robust_recovery}. The dashed blue line computed by Eq.\eqref{eq:noise1} and Eq. \eqref{eq:noise2} in Theorem \ref{thm:robust_recovery} is the theoretical prediction. The dashed black line is computed by Proposition \ref{pro:1}, and the red markers are the experimental results. We observe that the theoretical results and experimental results are closely matched, and the phase transition location is accurately predicted by Proposition \ref{pro:1}.
\subsection{Computation and Estimation  Trade-offs}
We shall verify Proposition \ref{pro:2} under the same settings as  Fig. \ref{fig:res} (a) via simulations. We fix the activity device number  as $ |\mathcal{S}| = 10$ to show the impact on the exact recovery using different smoothing parameters $ \mu $. As shown in Fig. \ref{fig:smooth}, the theoretical curve from Theorem \ref{thm:ptp3} closely matches the empirical curve of the $50\%$ success probability. Furthermore, it can be seen from Fig. \ref{fig:smooth} that increasing the smoothing parameter will result in  a larger  statistical dimension of the descent cone of $ \tilde{\mathcal{R}}(\bm\Theta)$, yielding longer signature sequences for signal recovery.   
\begin{figure}[ht]
        \centering
        \centerline{\includegraphics[width=0.95\columnwidth]{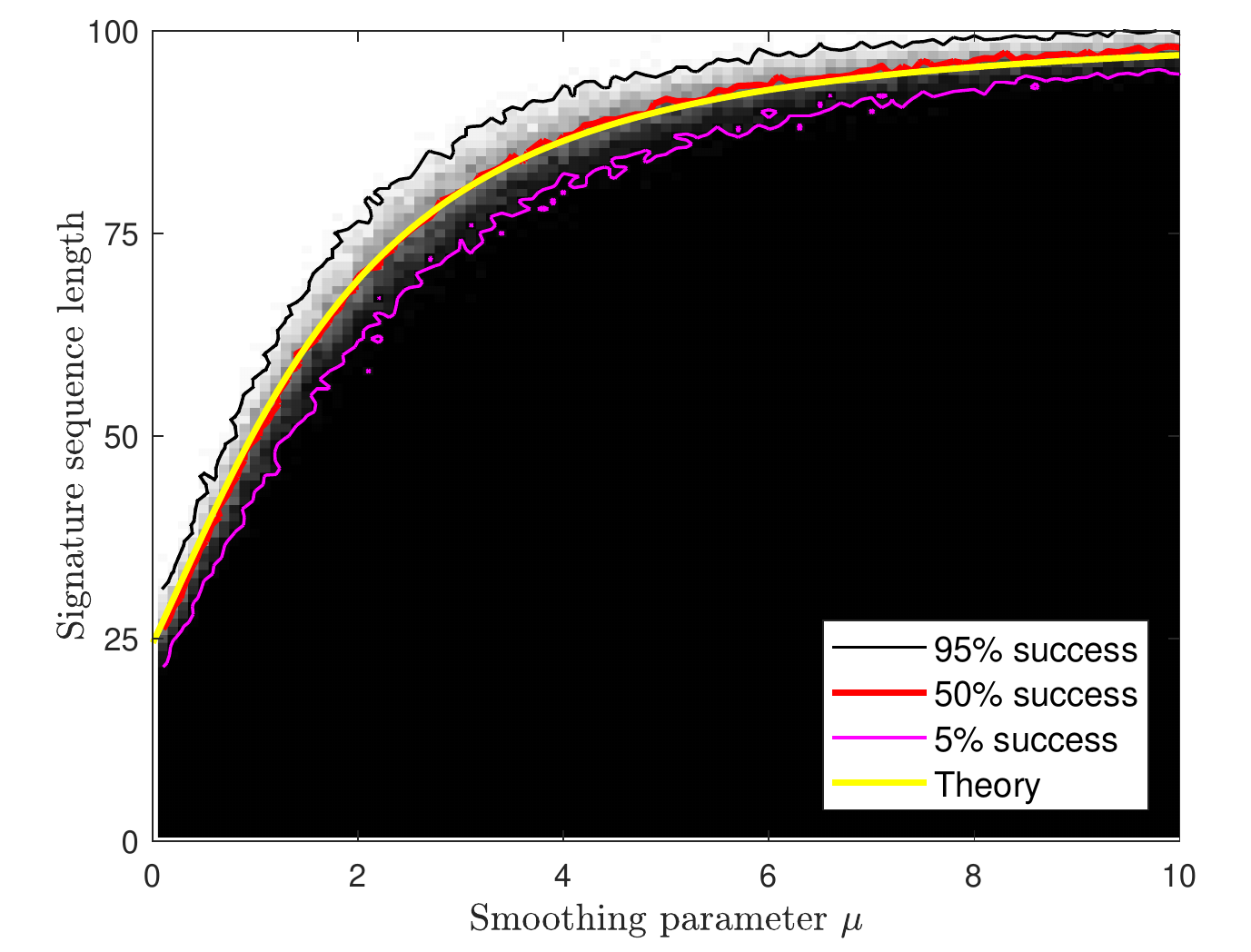}}
        \caption{Phase transitions in massive device connectivity via smoothing.}
        \label{fig:smooth}
\end{figure}

To evaluate the effectiveness of the smoothing method proposed in section \ref{sec:tradeoff}, we consider a scenario where the base station is equipped with $10$ antennas, and the total number of devices is set to be $2000$. We fix the number of active devices as $|\mathcal{S}|=100$.
For estimation problem $\mathscr{P}_s$, the channel matrix follows $\bm H\thicksim\mathcal{CN}(\bm 0,\bm I)$, the signature matrix follows $\bm Q\thicksim\mathcal{CN}(\bm 0,\bm I)$ and the additive noise matrix follows $\bm N\thicksim\mathcal{CN}(\bm 0,0.01\bm I)$. 

We compare the convergence behavior of Algorithm \ref{algo} with different amount of smoothing   under a fixed signature sequence length $ L=500 $ in Fig. \ref{fig:convergence_vs_mu}, which shows that increasing the amount of smoothing will accelerate the convergence rate significantly.
\begin{figure}[ht]
        \centering
        \centerline{\includegraphics[width=0.95\columnwidth]{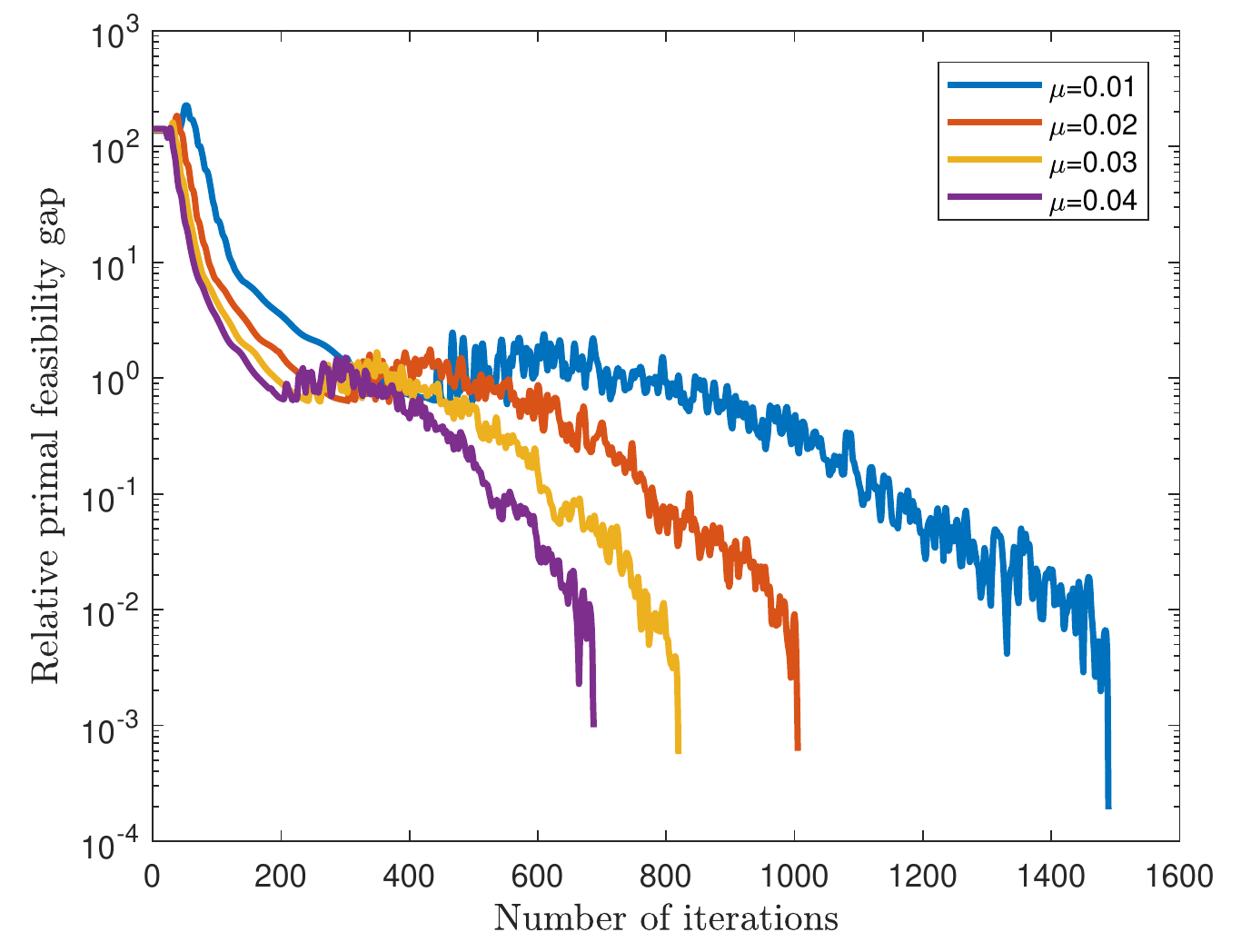}}
        \caption{Convergence behavior of Algorithm \ref{algo}.}
        \label{fig:convergence_vs_mu}
\end{figure}

Under a fixed signature sequence length $ L=500 $, we further solve problem  $\mathscr{P}_s$ using Algorithm \ref{algo} for different amount of smoothing $ \mu $ and stop it  when  $ \left|\norm{\bm{\tilde Q\tilde\Theta}-\bm{\tilde Y}}_F- \epsilon\right|/\epsilon \leq 10^{-3} $, where $ \epsilon $ is set according to \eqref{eq:eps}. The simulation results are averaged over $300$ times and are presented in Fig. \ref{fig:estimation_error}.  It can be seen that the average squared estimation error increases as the smoothing parameter $ \mu $ becomes large. This is because a larger smoothing parameter results in a larger statistical dimension $ \delta(\mathcal{D}(\tilde{\mathcal{R}}_G,\bm{\tilde\Theta}_0))$ as presented in Proposition \ref{pro:2}, and thus the estimation error increases as predicted by Theorem \ref{thm:robust_recovery}.  
\begin{figure}[t]
        \centering
        \centerline{\includegraphics[width=0.95\linewidth]{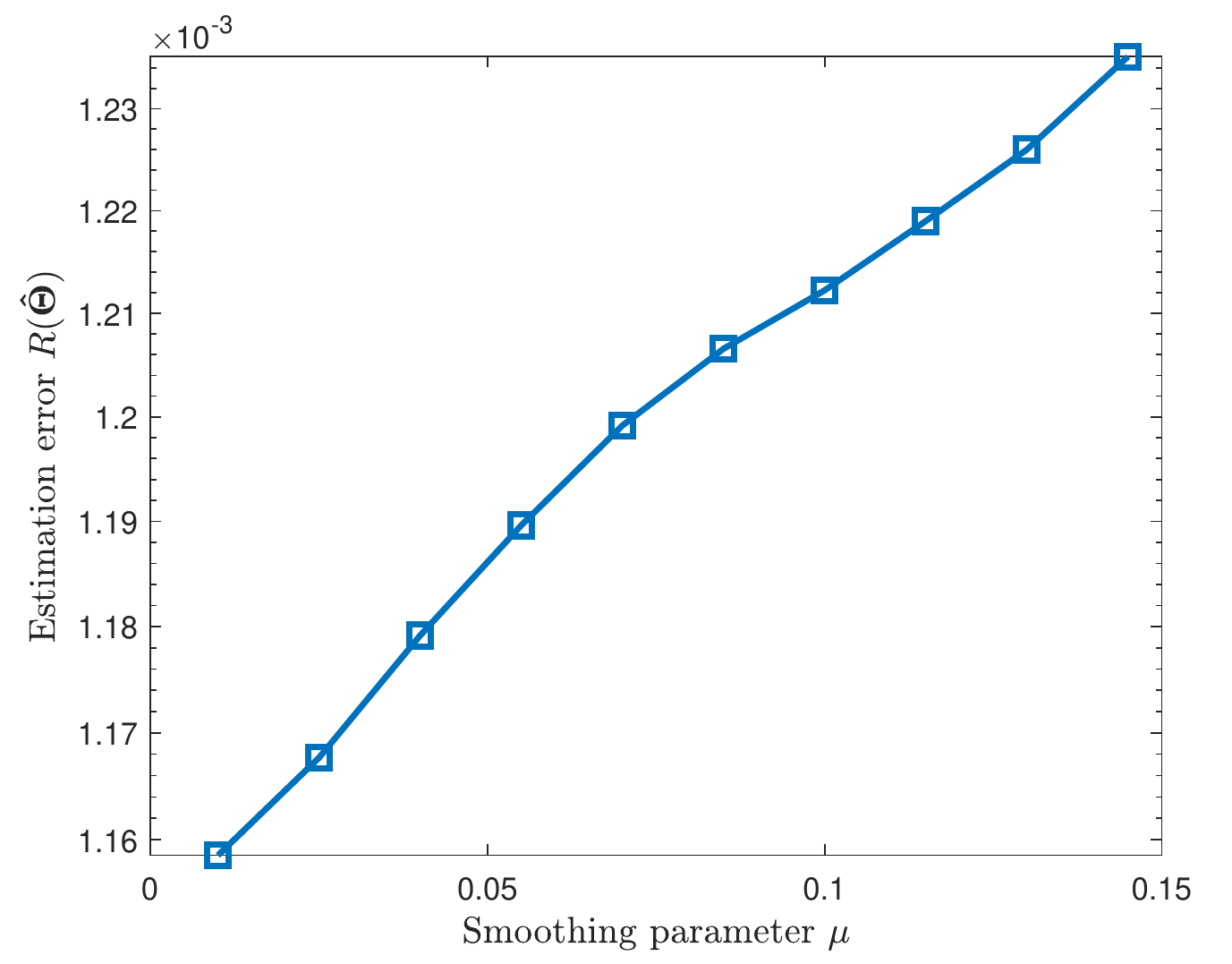}} 
        \caption{Estimation error versus smoothing parameter $ \mu $.}
        \label{fig:estimation_error}
\end{figure}
   
\section{Conclusions}
We developed a  structured group sparsity estimation approach to solve the joint active device detection and channel estimation problem in  IoT networks. Precise theoretical results were provided to characterize the location and width of the phase transition region for high-dimensional structured group sparsity estimation, which  provides theoretical guidelines for choosing the optimal signature sequence length. In particular, we observed that the transition width can be narrowed to zero asymptotically in the massive MIMO setting, yielding highly accurate phase transition location prediction. Numerical results demonstrated the accuracy of our developed theoretical results. Furthermore, we adopted the smoothing techniques to reduce the computational cost by accelerating the convergence rates, which yields a trade-off between computational cost and estimation accuracy. This was achieved by precisely characterizing the convergence rate in terms of smoothing parameter, signature sequence length and estimation accuracy via conic integral geometry. 
\appendices
\section{Proof of  theorem \ref{thm:robust_recovery}}\label{app:thm_robust_recovery}
\begin{proof}
        It turns out that the average squared prediction error satisfies \cite[ Lemma $7.1$ ]{oymak2016sharp},
        \begin{equation}\label{eq:noise3}
        2LM\cdot \underset{\sigma>0}{\max}~\frac{\mathbb{E}_{\bm{\tilde{N}}}[R(\bm{\tilde\Theta}^\ast)]}{\sigma^2}=2LM- \delta(\bm{\bar{Q}}\mathcal{C}(\bm{\tilde{\Theta}_0},\bm{\bar{Q}}^T)),
        \end{equation}
        where $\mathcal{C}(\bm{\tilde{\Theta}_0},\bm{\bar{Q}}^T)=\mathcal{D}
        (\mathcal{R}_G,\bm{\tilde{\Theta}_0})^{\circ}\cap \text{Range}(\bm{\bar{Q}}^T)$ with $\mathcal{D}(\mathcal{R}_G,\bm{\tilde{\Theta}_0})^{\circ}$ denoting the polar cone of $\mathcal{D}(\mathcal{R}_G,\bm{\tilde{\Theta}_0})$. Note that $\mathcal{C}(\bm{\tilde{\Theta}_0},\bm{\bar{Q}}^T)$  is the intersection of a cone with the uniformly random subspace, hence, Theorem \ref{thm:akf} is applicable. We split the problem into two cases.
        
        Whenever $L<\delta$, i.e., $2LM<\delta({\mathcal{D}(\mathcal{R}_G,\bm{\tilde{\Theta}_0})})$, using Theorem \ref{thm:akf}, we find with probability $1-c_1\exp(-c_2(L-\delta)^2/(NM^3))$, $\mathcal{D}
        (\mathcal{R}_G,\bm{\tilde{\Theta}_0})^{\circ}\cap \text{Range}(\bm{\bar{Q}}^T)=\{\bm 0\}$. We thus obtain $\delta(\bm{\bar{Q}}\mathcal{C}(\bm{\tilde{\Theta}_0},\bm{\bar{Q}}^T))=
        \delta(\bm{\bar{Q}}\{\bm 0\})=0$, yielding  \eqref{eq:noise1} by substituting it into  \eqref{eq:noise3}.
        
        Whenever $L>\delta$, i.e. $2LM>\delta({\mathcal{D}(\mathcal{R}_G,\bm{\tilde{\Theta}_0})})$, using the modification of Theorem \ref{thm:akf} which is given in  \cite[Proposition 13.1]{oymak2016sharp}, there exists constant $c_1,c_2>0$ such that with probability $1-c_1\exp(-c_2t^2)$, we have
        \begin{equation*}
        |\delta(\mathcal{C}(\bm{\tilde{\Theta}_0},\bm{\bar{Q}}^T))
        -(2LM-\delta({\mathcal{D}(\mathcal{R}_G,\bm{\tilde{\Theta}_0})}))|\leq t\sqrt{2MN}.
        \end{equation*}
        The rotational invariance property of statistical dimension gives
        \begin{equation*}
        \delta(\bm{\bar{Q}}\mathcal{C}(\bm{\tilde{\Theta}_0},\bm{\bar{Q}}^T))=
        \delta(\mathcal{C}(\bm{\tilde{\Theta}_0},\bm{\bar{Q}}^T)).
        \end{equation*}
        Consequently, we have
        \begin{equation*}
        |\delta(\bm{\bar{Q}}\mathcal{C}(\bm{\tilde{\Theta}_0},\bm{\bar{Q}}^T))
        -(2LM-\delta({\mathcal{D}(\mathcal{R}_G,\bm{\tilde{\Theta}_0})}))|\leq t\sqrt{2MN},
        \end{equation*}
        which  gives \eqref{eq:noise2} by combining with \eqref{eq:noise3}. 
        
        Using the fact \cite[Section 7.3]{oymak2016sharp} that
        \begin{equation}
        \underset{\sigma>0}{\max}~\frac{\mathbb{E}_{\bm{\tilde{N}}}[R(\bm{\tilde\Theta}^{\ast})]}{\sigma^2}+\underset{\sigma\rightarrow 0}{\lim}\frac{\mathbb{E}_{\bm{\tilde{N}}}[\hat{R}(\bm{\tilde{\Theta}}^\ast)]}{\sigma^2}=1,
        \end{equation} 
        gives \eqref{eq:noise11} and \eqref{eq:noise22}.
\end{proof}
   
\section{Proof of Proposition \ref{pro:1}}\label{app:proof_pro1}
\begin{proof}
        Since the regularizer $\mathcal{R}_G(\bm{\tilde\Theta}_0)$ is invariant under coordinate permutations of $\bm{\tilde\Theta}_0$. We assume without loss of generality that $\bm\Theta_0=[(\bm\theta_0^1)^{ T},\ldots,(\bm\theta_0^S)^{T},\bm 0_{M\times (N-S)}]^{ T}\in\mathbb{C}^{N\times M}$, where $\bm\theta_0^i$ are nonzero.
        Therefore, \eqref{eq:stdim_upbnd} becomes
        \begin{equation}
        \delta(\mathcal{D}(\mathcal{R}_G;\bm{\tilde\Theta_0}))\leq\underset{\tau\geq 0}{\inf}~\mathbb{E}
        [\mathrm{dist}^2(\bm G,\tau\cdot\partial\mathcal{R}_G(\bm{\tilde\Theta_0}))],
        \end{equation}
        where $\bm G\in\mathbb{R}^{2N\times M}$ has independent standard normal entries.
        
        The next step is to calculate $\partial\mathcal{R}_G(\bm{\tilde\Theta}_0)$. Assume $\bm Z\in\partial\mathcal{R}_G(\bm{\tilde\Theta}_0)$, then from the definition of the subdifferential, for any $\bm{\tilde\Theta}\in\mathbb{R}^{2N\times M}$ we have
        \begin{equation*}
        \sum_{i=1}^{N}\norm{\bm{\tilde\Theta}_{\mathcal{V}_i}}_F\geq \sum_{i=1}^{N}\norm{(\bm{\tilde\Theta}_0)_{\mathcal{V}_i}}_F+
        \sum_{i=1}^{N}\langle \bm Z_{\mathcal{V}_i}, \bm{\tilde\Theta}_{\mathcal{V}_i}-(\bm{\tilde\Theta}_0)_{\mathcal{V}_i} \rangle.
        \end{equation*}
        Specifically,  for some $\bm{\tilde\Theta}\in\mathbb{R}^{2N\times M}$ satisfying $\bm{\tilde\Theta}_{\mathcal{V}_j}=\bm 0$ for $j\neq i$, we have
        \begin{equation}\label{eq:subd}\small
        \forall \bm{\tilde\Theta}_{\mathcal{V}_i}\in\mathbb{R}^{2\times M}:\quad\norm{\bm{\tilde\Theta}_{\mathcal{V}_i}}_F\geq \norm{(\bm{\tilde\Theta}_0)_{\mathcal{V}_i}}_F+
        \langle \bm Z_{\mathcal{V}_i}, \bm{\tilde\Theta}_{\mathcal{V}_i}-(\bm{\tilde\Theta}_0)_{\mathcal{V}_i} \rangle,
        \end{equation}
        which means $\bm Z_{\mathcal{V}_i}\in\partial\norm{(\bm{\tilde\Theta}_0)_{\mathcal{V}_j}}_F$. Conversely, if \eqref{eq:subd} is satisfied for all $i\in[N]$, then summing over all indices shows that $\bm Z\in\partial\mathcal{R}_G(\bm{\tilde\Theta}_0)$. Hence,
        \begin{equation}
        \begin{split}
        &\bm Z\in\partial\mathcal{R}_G(\bm{\tilde\Theta_0})\Leftrightarrow
        \forall j\in[N]: \bm Z_{\mathcal{V}_j}\in\partial\norm{(\bm{\tilde\Theta}_0)_{\mathcal{V}_j}}_F\\
        &\Leftrightarrow
        \begin{cases}
        \bm Z_{\mathcal{V}_j}=\frac{(\bm{\tilde\Theta}_0)_{\mathcal{V}_j}}{\norm{(\bm{\tilde\Theta}_0)_{\mathcal{V}_j}}_F}
        &\mbox{for $j=1,\ldots,S$}\\
        \norm{\bm Z_{\mathcal{V}_j}}_F\leq 1
        &\mbox{for $j=S+1,\ldots,N$}.
        \end{cases}
        \end{split}
        \end{equation}
        
        Then we have
        \begin{equation}\small
        \begin{split}
        &\mathrm{dist}^2(\bm G,\tau\cdot\partial\mathcal{R}_G(\bm{\tilde\Theta}_0))\\
        &=\sum_{i=1}^{S}\norm{\bm G_{\mathcal{V}_i}-\tau\frac{(\bm{\tilde\Theta}  _0)_{\mathcal{V}_i}}{\norm{(\bm{\tilde\Theta}_0)_{\mathcal{V}_i}}_F}}_F^2+\sum_{i=S+1}^{N}
        \underset{\norm{\bm Z_{\mathcal{V}_i}}\leq 1}{\rm inf}\norm{\bm G_{\mathcal{V}_i}-\tau\bm Z_{\mathcal{V}_i}}_F^2\\
        &=\sum_{i=1}^{S}\norm{\bm G_{\mathcal{V}_i}-\tau\frac{(\bm{\tilde\Theta}  _0)_{\mathcal{V}_i}}{\norm{(\bm{\tilde\Theta}_0)_{\mathcal{V}_i}}_F}}_F^2+\sum_{i=S+1}^{N}
        \max\{\norm{\bm G_{\mathcal{V}_i}}_2-\tau,0\}^2.
        \end{split}
        \end{equation}
        Since the entries of $\bm G$ are independent standard normal,
        it has $2M$ degrees of freedom. Taking the expectation over the
        Gaussian matrix $\bm G$ gives
        \begin{equation}\label{eq:expc1}
        \begin{split}
        &\mathbb{E}[\mathrm{dist}^2(\bm G,\tau\cdot\partial\mathcal{R}_G
        (\bm{\tilde\Theta}_0))]\\
        &=S(2M+\tau^2)+(N-S)\mathbb{E}[\max\{\norm{\bm G_{\mathcal{V}_i}}_2-\tau,0\}^2]\\
        &=S(2M+\tau^2)+(N-S)\frac{2^{1-M}}{\Gamma(M)}\int_{\tau}^{\infty}
        (u-\tau)^2u^{2M-1}e^{-\frac{u^2}{2}}\mathrm{d}u,
        \end{split}
        \end{equation}
        where the following equality is applied
        \begin{equation*}
        \begin{split}
        &\mathbb{E}[\max\{\norm{\bm G_{\mathcal{V}_i}}_2-\tau,0\}^2]
        =\mathbb{E}[\max\{u-\tau,0\}^2]\\
        &=\frac{2^{1-M}}{\Gamma(M)}\int_{\tau}^{\infty}    (u-\tau)^2u^{2M-1}e^{-\frac{u^2}{2}}\mathrm{d}u
        \end{split}
        \end{equation*}
        in which $\norm{\bm G_{\mathcal{V}_i}}_F$ is replaced by a chi distribution variable $u$.
        
        Let $\rho=S/N$ and take the infimum over $\tau\geq 0$ we complete the proof of
        \eqref{eq:stdim_Rg}. The way to show that $\tau^\star$ is the unique minimizer of the righthand side of \eqref{eq:stdim_Rg} is similar to that in \cite[Proposition 4.5]{amelunxen2014living}.
\end{proof}

\section{Proof of Proposition \ref{pro:2}}\label{app:proof_pro2}
\begin{proof}
        The proof is  similar to Proposition \ref{pro:1}.
        We assume without loss of generality that $\bm\Theta_0=[(\bm\theta_0^1)^{ T},\ldots,(\bm\theta_0^S)^{T},\bm 0_{M\times (N-S)}]^{ T}\in\mathbb{C}^{N\times M}$, where $\bm\theta_0^i$ are nonzero. Therefore, \eqref{eq:stdim_upbnd} becomes
        \begin{equation}
                \delta(\mathcal{D}(\tilde{\mathcal{R}}_G;\bm{\tilde\Theta}_0))\leq\underset{\tau\geq 0}{\inf}~\mathbb{E}
                [\mathrm{dist}^2(\bm G,\tau\cdot\partial\tilde{\mathcal{R}}_G(\bm{\tilde\Theta}_0))],
        \end{equation}
        where $\bm G\in\mathbb{R}^{2N\times M}$ has independent standard normal entries.
        
        Since $\partial\tilde{\mathcal{R}}_G(\bm{\tilde\Theta}_0)=
        \partial\mathcal{R}_G(\bm{\tilde\theta}_0)+\frac{\mu}{2}
        \partial\norm{\bm{\tilde\Theta}_0}_F^2$, we have
        \begin{equation}\small
                \begin{split}
                        &\bm U\in\partial\mathcal{R}_G(\bm{\tilde\Theta}_0)\Longleftrightarrow\\
                        &\begin{cases}
                                \bm U_{\mathcal{V}_j}=(\bm{\tilde\Theta}_0)_{\mathcal{V}_j}/\norm{(\bm{\tilde\Theta}_0)_{\mathcal{V}_j}}_F
                                +\mu(\bm{\tilde\Theta}_0)_{\mathcal{V}_j}
                                &\mbox{if $j=1,\ldots,S$},\\
                                \norm{\bm U_{\mathcal{V}_j}}_F\leq 1
                                &\mbox{if $j=S+1,\ldots,N$}.
                        \end{cases}
                \end{split}
        \end{equation}
        Hence,
        \begin{equation}
                \begin{split}
                        &\mathrm{dist}^2(\bm G,\tau\cdot\partial\tilde{\mathcal{R}}_G(\bm{\tilde\Theta}_0))\\
                        &=\sum_{i=1}^{S}\norm{\bm G_{\mathcal{V}_i}-\tau((\bm{\tilde\Theta}_0)
                                _{\mathcal{V}_i}/\norm{(\bm{\tilde\Theta}_0)_{\mathcal{V}_i}}_F+
                                \mu(\bm{\tilde\Theta}_0)_{\mathcal{V}_j})}_F^2\\
                        &+\sum_{i=S+1}^{N}
                        \max\{\norm{\bm G_{\mathcal{V}_i}}_2-\tau,0\}^2.
                \end{split}
        \end{equation}
        Since the entries of $\bm G$ are independent standard normal,
        $\norm{\bm G_{\mathcal{V}_i}}_F$ follows the chi distribution
        with $2M$ degrees of freedom. Taking the expectation over the
        Gaussian matrix $\bm G$ gives
        \begin{equation}\label{eq:expc1}
                \begin{split}
                        &\mathbb{E}[\mathrm{dist}^2(\bm G,\tau\cdot\partial\mathcal{R}_G
                        (\bm{\tilde\Theta}_0))]
                        =S(2M+\tau^2(1+2\mu\bar{a}+\mu^2\bar{b}))\\
                        &+(N-S)\frac{2^{1-M}}{\Gamma(M)}\int_{\tau}^{\infty}
                        (u-\tau)^2u^{2M-1}e^{-\frac{u^2}{2}}\mathrm{d}u,
                \end{split}
        \end{equation}
        where $\bar{a}=\frac{1}{S}\sum_{i=1}^{S}\norm{(\bm{\tilde\Theta}_0)_{\mathcal{V}_i}}_F$, and
        $\bar{b}=\frac{1}{S}\sum_{i=1}^{S}\norm{(\bm{\tilde\Theta}_0)_{\mathcal{V}_i}}_F^2$.
        Let $\rho=S/N$ and taking the infimum over $\tau\geq 0$ completes the proof of \eqref{eq:stadim3}.
        
        The way to show that $\tau^\star$ is the unique minimizer of the righthand
        side of \eqref{eq:stadim3} is similar to that in \cite[Proposition 4.5]{amelunxen2014living}
\end{proof}

\bibliographystyle{ieeetr}
\bibliography{refs}
        
\end{document}